\pdfoutput=1
\documentclass[a4paper,twocolumn,11pt,unpublished,aps]{quantumarticle}
\pdfoutput=1

\usepackage[T1]{fontenc}
\usepackage{microtype}
\usepackage{amsmath,amssymb,amsfonts}
\usepackage{mathtools}
\usepackage{bm}
\usepackage{physics}

\usepackage{graphicx}
\usepackage{xcolor}
\usepackage{booktabs}
\usepackage{enumitem}
\usepackage{float}
\usepackage{algorithm}
\usepackage{multirow}
\usepackage{algpseudocode}
\usepackage[numbers,sort&compress]{natbib}

\usepackage{amsthm}
\usepackage{hyperref}

\newtheorem{theorem}{Theorem}[section]
\newtheorem{proposition}[theorem]{Proposition}

\theoremstyle{definition}
\newtheorem{definition}[theorem]{Definition}
\newtheorem{assumption}[theorem]{Assumption}

\theoremstyle{remark}
\newtheorem{remark}[theorem]{Remark}

\hypersetup{
  colorlinks=true,
  linkcolor=blue,
  citecolor=blue,
  urlcolor=blue,
  pdftitle={Reliable Adaptive Stopping for Krylov-Shadow Quantum Fisher
    Information Estimation},
  pdfauthor={Erjie Liu and Yangshuai Wang},
  pdfsubject={Quantum journal manuscript}
}

\newcommand{\QFI}{\mathcal{F}}
\let\Tr\relax
\DeclareMathOperator{\Tr}{Tr}

\newcommand{\eps}{\varepsilon}

\begin{document}

\title{Reliable Adaptive Stopping for Krylov-Shadow Quantum Fisher
Information Estimation}

\author{Erjie Liu}
\affiliation{Department of Mathematics, National University of Singapore,
10 Lower Kent Ridge Road, 119076, Singapore}
\author{Yangshuai Wang}
\affiliation{Department of Mathematics, National University of Singapore,
10 Lower Kent Ridge Road, 119076, Singapore}
\date{}

\begin{abstract}
Scalable quantum Fisher information (QFI) estimation becomes actionable
when the numerical estimate is paired with a trustworthy stopping
decision.
Krylov-shadow QFI estimation has two resource directions: the Krylov
order sets the population resolution, whereas the sample count controls
statistical uncertainty at that resolution.
We show that treating these directions as one can produce false stops,
where a width-based rule reports a narrow interval around a biased
low-order estimate.
We turn adaptive stopping into a two-component reliability problem,
separating Krylov truncation from finite-sample uncertainty, and
introduce AKS-QFI, a component-aware stopping interface for
Krylov-shadow estimators.
On a noisy mixed-state benchmark at $n=4$ qubits, width-only stopping
has false-stop rates from $0.16$ to $0.68$.
Under the same resource limit, AKS-QFI returns no false success
declarations; after recalibrating Krylov resolution and sample counts,
it returns accurate success declarations at true 5\% relative tolerance.
These results make adaptive stopping a reliability layer for
shadow-based QFI estimation.
\end{abstract}

\maketitle

\section{Introduction}
\label{sec:introduction}

Quantum Fisher information (QFI) sets the local precision limit in
quantum sensing and metrology~\cite{degen2017quantum,
giovannetti2011advances,taylor2016quantum}.
For a parameterized state $\rho_\theta$, the quantum
Cram\'{e}r--Rao bound lower bounds the mean-squared error of any
locally unbiased estimator by
$1/[N\mathcal{F}(\rho_\theta)]$ after $N$ measurement
repetitions~\cite{liu2020quantum, petz2011introduction,
jiang2014quantum, rath2021quantum, meyer2021fisher}.
Estimating QFI is therefore a basic validation step for
quantum-enhanced sensing protocols.
For scalable metrology, the output must be more than a number: it must
also state whether the evidence supports using that number to compare
protocols or certify a claimed precision advantage.
Exact evaluation, however, becomes exponentially costly as the
Hilbert space grows.
This has motivated randomized-measurement and shadow-based
estimators~\cite{beckey2022variational, vitale2024robust,
huang2020predicting, elben2022toolbox, zhang2025krylovshadow}
that replace exact diagonalization by controlled statistical
estimation.

We study the decision layer that turns a scalable QFI estimate into a
trustworthy success declaration.
An adaptive QFI routine must decide how to update its resources and,
more importantly, when its estimate is accurate enough to carry a
success label.
For Krylov-shadow QFI estimation this stopping decision has two
independent directions.
The Krylov order $K$ controls the resolution of the projected
symmetric-logarithmic-derivative problem, whereas the sample count $M$
controls finite-sample uncertainty at the current order.
More samples can sharpen the estimate around the current
Krylov-truncated target, while the low-$K$ projection bias remains.
Thus an adaptive routine can appear statistically converged while its
reported value is still systematically wrong.

We call this event a \emph{false stop}: the routine declares success,
but post hoc comparison with a reference value shows that the reported
estimate is outside tolerance.
This effect is quantitatively large in our benchmark.
A representative width-based run terminates at $(K,M)=(2,32)$ with
$\varepsilon=0.2$, while its final error is more than ten times larger
than $\varepsilon$.
Across the five-level dephasing grid in
Section~\ref{sec:benchmark-study}, the same rule has false-stop rates
between $0.16$ and $0.68$.
These observations isolate the central point of this work: scalable QFI
estimation needs a stopping rule that can distinguish sampling
convergence from Krylov-resolution convergence.

\subsection*{Main contributions}

\paragraph{A reliability problem for adaptive QFI estimation.}
We formulate false stopping as a distinct correctness problem for
adaptive Krylov-shadow QFI estimation.
The key observation is that the total error separates into a
Krylov-truncation component and a finite-sample component.
The two components are controlled by different resources and can become
inaccurate independently.
This is why a small bootstrap width at fixed $K$ is not a certificate
of total QFI accuracy.

\paragraph{A certifiable two-component contract.}
Theorem~\ref{thm:certified-stop} gives a sufficient condition for a
reliable success declaration when calibrated truncation and sampling
radii together fit inside the user tolerance.
Proposition~\ref{prop:finite-horizon-heldout} gives the adaptive
finite-horizon version: exploration data may nominate candidate stops,
but the success label is issued only after an independent confirmation
batch passes a truncation-plus-statistical radius test with alpha
spending.
This separates the certifiable contract from the empirical same-sample
checks used in the numerical rule.

\paragraph{A component-aware adaptive interface.}
We implement the principle as AKS-QFI (Adaptive Krylov-Shadow QFI), a
decision layer around the fixed-resource Krylov-shadow estimator.
AKS-QFI records separate resolution-side and sampling-side evidence
along the adaptive trajectory and permits a success declaration only
after minimum Krylov resolution, minimum sample count, and
persistence gates have passed.
The interface separates two outputs that are often conflated: a
numerical QFI estimate and the labeled decision that the estimate is
accurate enough to report as a success.

\paragraph{Benchmark evidence.}
On a noisy mixed-state benchmark with five dephasing levels and
50 independent replicates per noise level at $n=4$, the
component-aware empirical rule has no observed false stops, while the
width-based rule has false-stop rates from $0.16$ to $0.68$.
Under the default resource limit the component-aware rule makes no
false success declarations.
A true 5\% relative-tolerance control gives the positive case: after
Krylov resolution and sample counts are recalibrated,
component-aware stopping declares success in 12 of 20 runs with no
observed false stops, and all 20 component-aware terminal estimates are
within tolerance.

Shadow-based protocols for scalable observable estimation originate
with shadow tomography~\cite{aaronson2020shadow} and the
classical-shadow formalism~\cite{huang2020predicting}, and are
surveyed in the randomized-measurement
toolbox~\cite{elben2022toolbox}.
Krylov-subspace shadow methods extend these tools to QFI estimation by
introducing the explicit resolution knob
$K$~\cite{zhang2025krylovshadow}, while resampling-based uncertainty
assessment for shadow estimators has been explored
independently~\cite{ghysels2025shadowbootstrap}.
AKS-QFI adds a reliability layer for adaptive use: it decides whether an
estimate can be reported as accurate while both Krylov resolution and
sampling uncertainty are still changing.

The same distinction between apparent convergence and justified
accuracy appears broadly in finite-shot adaptive quantum algorithms,
but we focus on the QFI stopping problem because the two resource
directions are especially transparent.
Classical sequential analysis~\cite{wald1945sequential} and
anytime-valid inference~\cite{howard2021anytime} provide powerful
tools for stopping under uncertainty, including confidence sequences
valid along adaptive trajectories.
Standard applications most directly control a single martingale-type
error process.
Krylov-shadow QFI estimation instead requires joint calibration of a
truncation component and a sampling component, which shrink at
different rates and are governed by different resources.
Related adaptive stopping ideas have also been studied in quantum
state tomography~\cite{ferrie2014selfguided,granade2012rohl} and
adaptive resource allocation~\cite{lattimore2020bandit,
even2002pac}.
Posterior-width stopping in tomography does not address the same
failure mode when a narrow uncertainty interval coexists with
systematic Krylov-truncation bias.
Best-arm identification supports a decision after enough samples, but
AKS-QFI must also establish that the Krylov resolution has become
adequate, a condition separate from sampling statistics alone.

The remainder of the paper is organized as follows.
Section~\ref{sec:prelim} reviews QFI and the fixed-$(K,M)$
Krylov-shadow estimation setting.
Section~\ref{sec:aks} introduces AKS-QFI and its two-component
indicator.
Section~\ref{sec:stopping} presents the width-only and component-aware
stopping rules, states the main theoretical results, and
formalizes the false-stop pathology.
Section~\ref{sec:benchmark-study} reports benchmark evidence and
reliability plots.
Section~\ref{sec:discussion} discusses broader implications,
recalibration requirements, and next steps.
Appendices collect proofs, algorithmic details, benchmark
configuration, and additional checks.

\section{Preliminaries and Problem Setting}
\label{sec:prelim}

\subsection{Quantum Fisher information}

Let $\{\rho_\theta\}_{\theta\in\Theta}$ be a differentiable
one-parameter family of density operators on a
finite-dimensional Hilbert space $\mathcal{H}$.
The quantum Fisher information (QFI) at parameter value
$\theta$ is defined via the symmetric logarithmic derivative
(SLD) $L_\theta$, which is the self-adjoint operator
satisfying
\begin{equation}
\label{eq:sld}
\partial_\theta \rho_\theta
=
\frac{1}{2}\bigl(\rho_\theta L_\theta + L_\theta \rho_\theta\bigr),
\end{equation}
and
\begin{equation}
\label{eq:qfi-def}
\QFI(\rho_\theta) = \Tr\!\left(\rho_\theta L_\theta^2\right).
\end{equation}
For a mixed state with spectral decomposition
$\rho_\theta = \sum_k \lambda_k \ket{k}\!\bra{k}$,
the QFI admits the explicit formula
\begin{equation}
\label{eq:qfi-spectral}
\QFI(\rho_\theta)
= 2\sum_{\substack{j,k \\ \lambda_j+\lambda_k>0}}
\frac{\bigl|\bra{j}\partial_\theta\rho_\theta\ket{k}\bigr|^2}
     {\lambda_j+\lambda_k}.
\end{equation}
In the pure-state limit $\rho_\theta = \ket{\psi}\!\bra{\psi}$,
Eq.~\eqref{eq:qfi-spectral} reduces to
\begin{equation}
\label{eq:qfi-pure}
\QFI(\rho_\theta)
= 4\Bigl(
    \langle\partial_\theta\psi|\partial_\theta\psi\rangle
    - \bigl|\langle\psi|\partial_\theta\psi\rangle\bigr|^2
  \Bigr).
\end{equation}
Equation~\eqref{eq:qfi-spectral} serves as the reference
formula in our numerical evaluations: benchmark values
$\QFI_{\mathrm{ref}}$ are computed exactly via eigendecomposition
and are used only post hoc for evaluation; they are not
available to the adaptive procedure during a run.

\subsection{Krylov-shadow estimation with two resource knobs}

Reference~\cite{zhang2025krylovshadow} provides the
fixed-resource primitive on which this work builds.
For a Krylov order $K$ and sample count $M$ chosen in advance, the
Krylov-shadow estimator returns a fixed-resource QFI estimate.
The present work does not change this primitive.
It studies the adaptive layer above it: a user specifies an accuracy
target, and the routine must decide how far to increase $K$, how many
randomized measurements to use, and whether the evidence supports a
success label.
The key point is that the two resources affect different targets.
The Krylov order changes the population value approached by the
estimator, whereas the sample count controls finite-sample error
around that fixed-$K$ population value.

At a fixed Krylov order, the estimator replaces the full QFI
calculation by a projected calculation on
\[
\mathcal{K}_K(G,v_0)=
\mathrm{span}\{v_0,Gv_0,\ldots,G^{K-1}v_0\},
\]
where $G$ is the Hermitian generator of the unitary encoding and
$v_0$ is the seed vector.
The density matrix and generator are projected to this subspace;
the projected density matrix is normalized, the mixed-state QFI
formula is evaluated on the projected pair, and the result is
weighted by the retained trace of the projected density matrix.
The resulting population value is denoted by $\QFI_K$.

Randomized measurements enter through a classical-shadow density
estimate,
\begin{equation}
\label{eq:ks-estimator}
\widehat{\rho}_M
= \frac{1}{M}\sum_{m=1}^{M}
  \mathcal{M}^{-1}\!\left(U_m
  |b_m\rangle\!\langle b_m|U_m^\dagger\right),
\end{equation}
where $\mathcal{M}^{-1}$ is the single-shot classical-shadow inverse
channel~\cite{huang2020predicting}, $U_m$ is the random measurement
basis, and $b_m$ is the observed bit string.
The fixed-resource QFI estimate is then
\[
\widehat{\QFI}_{K,M}=\Phi_K(\widehat{\rho}_M,G),
\]
where $\Phi_K$ denotes the projected mixed-state QFI functional.
In the benchmark estimator, $\widehat{\rho}_M$ is first projected
back to the density-matrix cone before $\Phi_K$ is evaluated.
Because $\Phi_K$ is nonlinear, this plug-in estimate need not be
exactly unbiased at finite $M$; any such finite-$M$ effect belongs to
the sampling side of the error decomposition.

This construction exposes the error split used throughout the paper.
The large-sample target at fixed $K$ is $\QFI_K$, while the desired
target is the full QFI $\QFI(\rho_\theta)$.
The total error therefore separates into a resolution-side component
$\QFI_K-\QFI(\rho_\theta)$ and a sampling-side component
$\widehat{\QFI}_{K,M}-\QFI_K$, as formalized in
Appendix~\ref{sec:theory-decomp}.
Increasing $K$ acts on the truncation component; increasing $M$
reduces finite-$M$ sampling error, including random fluctuation and
any plug-in bias at fixed $K$~\cite{elben2022toolbox}.
These operations are not interchangeable.
More samples can sharpen the estimate around the current $\QFI_K$,
but they cannot by themselves remove low-$K$ projection bias.

The adaptive interface therefore needs evidence tied to both error
sources.
We denote by $d_K$ a \emph{Krylov-stability measure}, implemented here
as an inter-order change
$|\widehat{\QFI}_{K,M} - \widehat{\QFI}_{K-1,M}|$, and by $w_M$ the
bootstrap empirical uncertainty interval width at the current
$(K,M)$.
The role of $w_M$ is deliberately narrow: it measures finite-sample
spread at the current Krylov order.
It does not test whether the Krylov subspace has resolved the
QFI-relevant directions, and therefore it does not by itself bound the
total QFI error.
The implementation-level definitions of $(d_K,w_M)$ are given in
Appendix~\ref{app:impl}.

\subsection{Problem statement: adaptive stopping}

Given a target tolerance $\varepsilon > 0$ and an optional risk
parameter $\delta \in (0,1)$, an adaptive QFI routine must return
three pieces of information: a terminal resource pair $(K,M)$, a
terminal estimate $\widehat{\QFI}_{K,M}$, and a labeled stopping
outcome.
The outcome label is part of the contract.
A success declaration should mean that the available evidence supports
the claimed tolerance.
A resource-limit termination should mean the opposite: within the
resource limit, the routine did not obtain enough evidence to make
that claim.

The central difficulty is that the quantities available online are
observable proxies rather than calibrated total-error radii.
The sampling-side width $w_M$ can become small at low $K$, and a local
inter-order change can underestimate unresolved Krylov bias even when
the true truncation bias $B_K$ remains large.
Thus a stopping rule without minimum-resource and persistence gates
can terminate prematurely in noisy mixed-state regimes.
Addressing this failure mode requires two design choices: an indicator
that keeps resolution-side and sampling-side evidence separate, and
a stopping rule that treats a success declaration as a separate
decision rather than as a by-product of a narrow empirical interval.
Section~\ref{sec:aks} gives this interface, and
Section~\ref{sec:stopping} states the component conditions under which
a stopping claim can be trusted.

\section{AKS-QFI: Adaptive Interface and Error Indicator}
\label{sec:aks}

AKS-QFI turns a fixed-$(K,M)$ Krylov-shadow QFI estimator into a
reliability-aware adaptive interface.
Instead of treating the Krylov order $K$ and the sample count $M$
as offline choices, the interface selects them along an adaptive
trajectory under a user tolerance $\varepsilon$.
Its main role is not to modify the estimator formula, but to separate
two outputs that are often conflated: the numerical QFI estimate and
the decision that this estimate is supported at the requested
tolerance.
The fixed-resource Krylov-shadow estimator remains the primitive;
AKS-QFI defines the decision contract around its use.

\subsection{Inputs, outputs, and the stopping interface}

AKS-QFI takes as input a target tolerance $\varepsilon > 0$,
resource limits $K \le K_{\max}$ and $M \le M_{\max}$, and
optional configuration parameters for the stopping protocol.
At termination it returns a point estimate $\widehat{\QFI}$
together with a concise stopping summary.
The summary includes the final $(K, M)$, the empirical interval
width, the Krylov-stability value, and an outcome label: either a
success declaration or resource-limit termination.
This label is not auxiliary metadata.
It is the output that tells a downstream user whether the adaptive run
supports the requested tolerance claim or reports the best value
found within the available resources.

\subsection{Observable indicator and allocation rule}
\label{sec:aks-indicator}

Because the total estimation error has two qualitatively
distinct components, a stopping rule based on a single empirical
width can be misleading.
The finite-$M$ sampling error may be small while the
Krylov-truncation error remains dominant.
AKS-QFI therefore maintains two observable component scores and uses
their maximum only as an allocation severity score,
\begin{equation}
\label{eq:indicator}
\widetilde I_{K,M}
= \max\{\widetilde I^{(\mathrm{trunc})}_{K,M},
        \widetilde I^{(\mathrm{stat})}_{K,M}\}.
\end{equation}
The maximum in Eq.~\eqref{eq:indicator} routes the next resource
update; it is not itself a certificate.
The stopping test remains componentwise and resource-aware.
The sampling side must pass its empirical width gate.
The Krylov side can pass by local stability, when the inter-order
change is below tolerance.
Under the finite-resource convention used in the numerical rule,
reaching the prescribed Krylov limit $K_{\max}$ also passes the
empirical Krylov-side gate; this means only that the run has used all
allowed Krylov resolution, not that a calibrated truncation radius has
been proved.
Theorem~\ref{thm:certified-stop} gives the corresponding
assumption-explicit statement only when calibrated component bounds,
not raw observable scores, satisfy the half-tolerance conditions.
The two observable quantities are defined as follows.

The \emph{truncation-side term} is derived from the
Krylov-stability measure $d_K$ introduced in
Section~\ref{sec:prelim}:
\begin{equation}
\label{eq:I-trunc}
\widetilde I^{(\mathrm{trunc})}_{K,M}
\coloneqq d_K
= \bigl|\widehat{\QFI}_{K,M} - \widehat{\QFI}_{K-1,M}\bigr|,
\end{equation}
which measures local Krylov stability through the inter-order
change in the point estimate.
This comparison is evaluated at the same sample count $M$ so that
changes in the score primarily reflect the Krylov order rather
than a changed shot count.
The auxiliary $(K-1,M)$ evaluation does not use bootstrap resampling.

The \emph{sampling-side observable quantity} is derived from a
bootstrap empirical uncertainty interval constructed from the
$M$ shadow samples.
Concretely, let $[\widehat{L}_{K,M}, \widehat{U}_{K,M}]$
denote the equal-tailed bootstrap interval at nominal level
$1-\alpha$; then
\begin{equation}
\label{eq:I-stat}
\widetilde I^{(\mathrm{stat})}_{K,M}
\coloneqq W^{(\mathrm{stat})}_{K,M}
\coloneqq w_M
= \widehat{U}_{K,M} - \widehat{L}_{K,M}.
\end{equation}
As established in Section~\ref{sec:prelim}, $w_M$ is treated
as a finite-sample width at fixed $K$.
It is not a bound on the total QFI error unless the Krylov truncation
term is also controlled and the interval coverage is justified under
explicit assumptions.
The calibrated sampling term in Theorem~\ref{thm:certified-stop}
is an error radius; $w_M$ is the full bootstrap interval width used by
the empirical stopping rule.
Thus the gate $w_M \le \varepsilon$ is a tolerance-scale check on the
empirical interval width; the theorem uses a calibrated error radius.
The tildes mark observable decision quantities and separate them from
the calibrated component bounds used in
Theorem~\ref{thm:certified-stop}.
The theoretical status of this interval is analyzed in
Section~\ref{sec:stopping}.

\paragraph{Allocation rule.}
The allocation rule gives priority to the resource associated with the
unresolved component.
At each iteration, AKS-QFI first checks whether the Krylov-stability
measure still exceeds tolerance.
If so, it increments the Krylov order
($K \leftarrow K+1$).
If the Krylov side has passed but the sampling interval width remains
above tolerance, it doubles the sample count
($M \leftarrow 2M$).
Geometric growth in $M$ matches the expected $1/\sqrt{M}$ scaling of
the statistical term and keeps the number of adaptive iterations
small.
Complete pseudocode, including tie-breaking and final-pass
details, is given in Algorithm~\ref{alg:aks-qfi} of
Appendix~\ref{app:impl}.
The rule is intended as an interpretable finite-resource controller, not
as an optimal allocation algorithm.

\subsection{Termination summary}

The terminal summary is the object evaluated in the benchmark.
It records the final estimate, final resource pair, empirical interval
width, local Krylov-stability value, and stopping outcome.
Post hoc access to $\QFI_{\mathrm{ref}}$ then separates three cases
that would otherwise look similar from the estimate alone: an accurate
success declaration, a false success declaration, and a run that
reaches the resource limit without a success label.

\subsection{Measurement-count accounting}
\label{sec:aks-complexity}

Each invocation of the Krylov-shadow estimator at order $K$
and sample count $M$ uses $M$ randomized measurements.
Within a single AKS-QFI run, $K$ increases by at most
$K_{\max} - K_0$ steps and $M$ doubles at most
$L \coloneqq \lceil\log_2(M_{\max}/M_0)\rceil$ times.
Thus the number of estimator evaluations satisfies
\[
N_{\mathrm{eval}}
\le (K_{\max}-K_0)+L+1.
\]
An implementation that discards previous measurements and
redraws at every evaluation would therefore obey
\begin{equation}
\label{eq:measurement-count-bound}
N_{\mathrm{meas}}^{\mathrm{fresh}}
\;\le\;
N_{\mathrm{eval}}\,M_{\max}.
\end{equation}
With the defaults ($K_{\max}=8$, $K_0=1$, $M_0=16$,
$M_{\max}=512$), this gives $N_{\mathrm{eval}}\le 14$.
This fresh-batch upper bound treats every intermediate evaluation as a
new batch.
An online implementation can instead acquire measurements lazily and
reuse nested prefixes, in which case the number of shots entering the
terminal estimate is $M_{\mathrm{final}}$.

\textit{Effective sample-count convention.}
The reported summaries use the \emph{terminal sample count}
$M_{\mathrm{final}}$, namely the number of samples entering
the terminal estimate and stopping decision.
This convention measures how early the decision layer would have
stopped under nested online acquisition and is therefore the relevant
quantity for comparing stopping rules in this paper.
Equation~\eqref{eq:measurement-count-bound} gives a fresh-batch upper
bound for a protocol that instead redraws fresh data at every
evaluation.

\section{Stopping Rules and the False-Stop Failure Mode}
\label{sec:stopping}

This section turns the outcome label returned by AKS-QFI into a
formal reliability question.
A success declaration is a statistical claim about the terminal
estimate, and in the Krylov-shadow setting that claim has two
components.
It must be supported by resolution-side evidence and by sampling-side
evidence; a narrow empirical interval addresses only the second of
these.
We first describe the width-based failure mode, then state the
component condition that would justify a success declaration, and
finally introduce the component-aware empirical rule used in the benchmark.

\subsection{Width-only stopping and the low-\texorpdfstring{$K$}{K}
failure mode}

A natural baseline terminates when the observable severity score from
Section~\ref{sec:aks-indicator} falls below the user tolerance,
\begin{equation}
\label{eq:width-only-stop}
\widetilde I_{K,M} \le \varepsilon.
\end{equation}
This rule is simple and instance-adaptive, but it can conflate local
numerical stability with total accuracy.
A narrow empirical interval indicates internal consistency of the
observed samples, and a small adjacent-order change indicates local
stability between two nearby Krylov orders.
Neither statement establishes that the fixed-$K$ population value is
close to the true QFI\@.
In noisy mixed-state regimes, the estimator can therefore stabilize
around a biased low-$K$ value while both observable scores appear
small, leading to a premature success declaration.

The corresponding reliable-stopping condition is componentwise.
In the theorem, $I^{(\mathrm{trunc})}_{K}$ and
$I^{(\mathrm{stat})}_{K,M}$ denote calibrated upper bounds
on the truncation and sampling errors, respectively.
These bounds are stronger objects than the observable scores
$(d_K,w_M)$: they are error radii with stated assumptions and risk
level.
The theorem says that a success declaration is justified when both
components are controlled at half-tolerance scale.
The empirical false-stop analysis below then asks what happens when a
practical rule substitutes observable scores for those calibrated
radii.

\begin{theorem}[Sufficient two-component stopping condition]
\label{thm:certified-stop}
Under Assumptions~\ref{ass:krylov-conv}
and~\ref{ass:subgauss} (stated in
Appendix~\ref{app:theory}), let $(K,M)$ be a fixed resource
pair, or an adaptively selected terminal pair for which the
sampling-side bound is valid at the terminal time.
For any $\varepsilon > 0$ and $\delta \in (0,1)$, if
\begin{align}
\label{eq:suff-trunc-main}
I^{(\mathrm{trunc})}_{K} &\le \frac{\varepsilon}{2}, \\
\label{eq:suff-stat-main}
I^{(\mathrm{stat})}_{K, M} &\le \frac{\varepsilon}{2},
\end{align}
then the corresponding success claim is reliable at level
$1-\delta$:
\begin{equation}
\label{eq:certified-main}
\Pr\!\left[
  \bigl|\widehat{\QFI}_{K,M} - \QFI(\rho_\theta)\bigr|
  > \varepsilon
\right] \le \delta.
\end{equation}
\end{theorem}

\begin{proof}
See Appendix~\ref{sec:theory-stopping}.
\end{proof}

Theorem~\ref{thm:certified-stop} is a component condition at the
terminal resource pair.
For an adaptive interface, a second issue appears: the terminal pair is
chosen from the data.
The held-out construction below separates candidate selection from the
final success declaration.

\begin{proposition}[Finite-horizon held-out certificate]
\label{prop:finite-horizon-heldout}
Consider an adaptive run with at most $J$ eligible success attempts.
Before attempt $j$, let $\mathcal{H}_j$ denote the exploration
history used to choose the candidate pair $(K_j,M_j)$.
At that attempt, draw a confirmation batch independent of
$\mathcal{H}_j$ and form the reported estimate
$\widehat{\QFI}^{\,\mathrm{conf}}_j$.
Assume a pre-calibrated Krylov truncation radius
$R^{(\mathrm{trunc})}_{K_j}$ and a confirmation-batch statistical
radius $R^{(\mathrm{stat})}_{j}(\delta_j)$ such that, conditional on
the exploration history,
\begin{multline}
\Pr\!\Bigl[
  \bigl|\widehat{\QFI}^{\,\mathrm{conf}}_j-\QFI\bigr|
  >
  R^{(\mathrm{trunc})}_{K_j}
  + R^{(\mathrm{stat})}_{j}(\delta_j)
\\
  \,\Bigm|\, \mathcal{H}_j
\Bigr]\le \delta_j .
\end{multline}
If the routine declares success at attempt $j$ only when
\begin{equation}
\label{eq:heldout-cert-main}
R^{(\mathrm{trunc})}_{K_j}
  + R^{(\mathrm{stat})}_{j}(\delta_j)
\le \varepsilon,
\end{equation}
and if $\sum_{j=1}^{J}\delta_j\le\delta$, then
\begin{equation}
\begin{aligned}
\Pr\!\bigl[
\exists j\le J:\;&
  \text{success declared at }j
\\[-1pt]
  &\text{and }
  \left|\widehat{\QFI}^{\,\mathrm{conf}}_j-\QFI\right|>\varepsilon
\bigr]\le\delta.
\end{aligned}
\end{equation}
\end{proposition}

\begin{proof}
Let $A_j$ be the event that attempt $j$ declares success and the
confirmed estimate has error larger than $\varepsilon$.
On the success event,
Eq.~\eqref{eq:heldout-cert-main} holds, so $A_j$ is contained in
the conditional failure event controlled at level $\delta_j$.
Conditioning on $\mathcal{H}_j$ therefore gives
$\Pr(A_j\mid\mathcal{H}_j)\le\delta_j$.
Taking expectations gives $\Pr(A_j)\le\delta_j$.
A union bound over the at most $J$ attempts gives the result.
\end{proof}

In the held-out certificate interface, the adaptive
phase produces only a candidate stop.
The returned point estimate for a successful run is then recomputed on
the independent confirmation batch.
The statistical radius is the larger half-width of the held-out
bootstrap interval calibrated at the conditional level $1-\delta_j$,
\[
R^{(\mathrm{stat})}_{j}(\delta_j)
= \max\{
|\widehat{\QFI}^{\,\mathrm{conf}}_j-\widehat L_j|,
|\widehat U_j-\widehat{\QFI}^{\,\mathrm{conf}}_j|
\},
\]
with Bonferroni spending $\delta_j=\delta/J$ in the finite-horizon
experiments.
The symmetric radius contains the bootstrap interval, so calibrated
conditional coverage for $[\widehat L_j,\widehat U_j]$ implies the
radius condition assumed in
Proposition~\ref{prop:finite-horizon-heldout}.
The truncation radius is read from a pre-registered Krylov calibration
table for the same family, generator, seed strategy, noise model, and
system size.
This distinction fixes the language used in the rest of the paper:
same-sample empirical gates are reported as observed reliability
under calibrated resources, whereas finite-horizon certificate refers
to the held-out interface with calibrated component radii.

Theorem~\ref{thm:certified-stop} and
Proposition~\ref{prop:finite-horizon-heldout} also clarify the role of
the empirical rule.
The observable quantities $(d_K,w_M)$ measure local Krylov stability
and sampling width; calibrated total-error control requires calibrated
component radii.
A single false stop is compatible with the allowed
$\delta$-probability failure event of a valid certificate.
Systematic false stops, or false stops produced by same-sample checks
that are not calibrated radii, mean that the empirical tests did not
establish those component conditions.

\subsection{Post hoc definition of a false stop}

When a reference value $\QFI_{\mathrm{ref}}$ is available for
benchmarking, the post hoc absolute error is
\begin{equation}
\label{eq:post-hoc-error}
E_{K,M} \coloneqq
\bigl|\widehat{\QFI}_{K,M} - \QFI_{\mathrm{ref}}\bigr|.
\end{equation}
A run exhibits a \emph{false stop} if it declares success while
\begin{equation}
\label{eq:false-stop-def}
E_{K,M} > \varepsilon.
\end{equation}
This definition is purely evaluative.
The reference value is unavailable to the adaptive routine and is used
only after the run to assess stopping reliability.

Theorem~\ref{thm:certified-stop} provides a componentwise
interpretation for false stops.

\begin{remark}[Componentwise interpretation of a false stop]
\label{rem:false-stop-nec}
Under calibrated component bounds satisfying the assumptions of
Theorem~\ref{thm:certified-stop}, if both half-tolerance
conditions hold at the terminal pair, a false success declaration can
occur only on the theorem's $\delta$-probability failure event.
Consequently, a false success outside that allowed event, or a
systematic false-success pattern in a benchmark, must be traced to at
least one component not being controlled at the claimed tolerance
scale, i.e.,
\begin{equation}
\label{eq:false-stop-nec}
I^{(\mathrm{trunc})}_{K} > \frac{\varepsilon}{2}
\quad \text{or} \quad
I^{(\mathrm{stat})}_{K,M} > \frac{\varepsilon}{2}.
\end{equation}
\end{remark}

In the false-stop instances documented in
Section~\ref{sec:benchmark-study}, the empirical stopping tests
pass at low resources, while post hoc comparison with the
reference value and higher-$K$ trajectories shows that the
truncation component remains the limiting source of error.

\subsection{False-stop behavior in noisy mixed-state regimes}

The noisy mixed-state family used in the benchmark is chosen to expose
this mechanism.
The characteristic pattern is that the empirical tests pass at low
resources even though increasing $K$ would still move the estimate
toward the reference value, consistent with
Remark~\ref{rem:false-stop-nec}.
A representative termination instance is provided in
Appendix~\ref{sec:failure_example}, and aggregate
false-stop rates across noise levels are reported in
Section~\ref{sec:benchmark-study}.

\subsection{Component-aware stopping: eligibility and persistence}

Component-aware stopping is designed to suppress low-resource false stops by
augmenting~\eqref{eq:width-only-stop} with an eligibility gate
and a persistence requirement.
The eligibility gate requires
\begin{equation}
\label{eq:stop-thresholds}
K \ge K_{\min}^{\mathrm{stop}},
\qquad
M \ge M_{\min}^{\mathrm{stop}},
\end{equation}
before a success declaration can be considered.
The persistence requirement then asks the resource-aware component gates
to pass for $P$ consecutive eligible iterations.
Equivalently, the empirical severity
$\widetilde I_{K,M}$ must remain below $\varepsilon$ for $P$ eligible
rounds, with the finite-resource convention that reaching $K_{\max}$
passes the Krylov-side gate.
That convention is not a calibrated half-tolerance condition of
Theorem~\ref{thm:certified-stop}; it records that no further Krylov
resolution is available under the resource limit.
The component-aware rule therefore acts as a decision filter around the
observable scores; distribution-free accuracy requires calibrated
component radii.

The thresholds in~\eqref{eq:stop-thresholds} are tied to the
two-component error model.
Appendix~\ref{sec:theory-stopping-params} derives closed-form
expressions for $K_{\min}^{\mathrm{stop}}(\varepsilon)$ and
$M_{\min}^{\mathrm{stop}}(\varepsilon, \delta)$ from the
error decomposition of Theorem~\ref{thm:certified-stop},
and Proposition~\ref{prop:patience} gives an idealized
independence model in which the false-stop probability at an
eligible step decreases geometrically with $P$.
In the benchmark, these expressions motivate the default
configuration
\begin{equation}
\label{eq:stop-defaults}
K_{\min}^{\mathrm{stop}} = 4,
\qquad
M_{\min}^{\mathrm{stop}} = 128,
\qquad
P = 2,
\end{equation}
for the tested noise family.
The sensitivity of stopping reliability to these parameters is studied
in Section~\ref{sec:threshold-ablation}.
Like AKS-QFI itself, component-aware stopping acts only on the decision layer;
the underlying estimator calls are unchanged.
Its reliability and cost tradeoffs are evaluated using the
metrics defined in Section~\ref{sec:benchmark-study}.

\section{Benchmark Study}
\label{sec:benchmark-study}

We evaluate the stopping layer of AKS-QFI on controlled benchmarks
that isolate the decision problem from the underlying estimator.
All comparisons keep the Krylov-shadow estimator fixed and change only
the stopping rule.
The chapter is organized around three operational questions.
First, can a width-only empirical rule declare success before the
two error components are controlled?
Second, does the component-aware rule suppress such false success
declarations under the same fixed resource limit?
Third, after Krylov resolution and sample count are recalibrated, can
the same component-aware principle produce correct success declarations
in addition to avoiding premature ones?
The evidence follows this order.
Figure~\ref{fig:width-vs-err} identifies the false-stop mechanism;
Figs.~\ref{fig:false-stop-noise} and~\ref{fig:median-err-noise}
show that the mechanism is systematic across the powered $n=4$ grid;
Fig.~\ref{fig:relative-recalibration-width-error} shows the positive
case after recalibration; and
Figs.~\ref{fig:threshold-ablation-pareto} and
\ref{fig:krylov-convergence} explain which thresholds and Krylov
calibration choices set the decision scale.

\subsection{Benchmark design and decision metrics}
\label{sec:exp-setup}

We compare width-only stopping~\eqref{eq:width-only-stop} against
component-aware empirical stopping with the default
configuration~\eqref{eq:stop-defaults} on the
noisy mixed-state family with depolarizing probability
$p_{\mathrm{dep}}=0.03$ and resource limits
$(K_{\max}, M_{\max}) = (8, 512)$.
Dephasing probability varies over
$p_\phi \in \{0, 0.06, 0.12, 0.18, 0.24\}$.
The main-text reliability benchmark is the $n=4$ qubit case.
The same protocol is also run at $n=6$ and $n=8$, but those data are
reported in Appendix~\ref{app:scaling} as calibration-transfer tests
rather than as a scaling study.

The primary $n=4$ benchmark uses a fixed absolute tolerance
$\varepsilon=0.2$.
This choice keeps the stopping rule independent of the post hoc
reference value and is used consistently in the threshold ablation and
representative trajectories.
The reference value $\QFI_{\mathrm{ref}}$ is computed exactly via
Eq.~\eqref{eq:qfi-spectral} only for post hoc evaluation of the
terminal error and false-stop status; it is not used by the stopping
logic.
The $n=4$ results (50 independent replicates per noise level)
constitute the primary benchmark and provide the statistical basis for
the false-success-suppression results reported in
Section~\ref{sec:exp-aggregate}.

The implementation also includes two sequential statistical baselines
and one held-out confirmation rule.
The \texttt{seq\_heldout\_width} baseline uses the width-only
candidate stop but requires an independent held-out bootstrap
confirmation with alpha spending; it does not use any Krylov
truncation radius.
The \texttt{fixedK\_heldout} baseline locks $K$ to a calibrated
fixed value and then applies the same held-out sampling-radius test.
The held-out component-aware rule combines the component-aware
K-stability candidate stop with the held-out finite-horizon certificate of
Proposition~\ref{prop:finite-horizon-heldout}.
Its success condition is
$R^{(\mathrm{trunc})}_{K}+R^{(\mathrm{stat})}_{\mathrm{conf}}
\le\varepsilon$, and its reported estimate is the held-out
confirmation estimate.
These rules define the certifiable interface studied theoretically in
Proposition~\ref{prop:finite-horizon-heldout}.
The numerical results below therefore separate the empirical
component-aware rule from lightweight recalibration controls; a powered
larger-system held-out study is the next calibration step.

The $n \in \{6,8\}$ calibration-transfer tests use a relative
tolerance
$\varepsilon_{\mathrm{rel}}=0.05$, i.e.,
$\varepsilon=\varepsilon_{\mathrm{rel}}\QFI_{\mathrm{ref}}$, to
normalize the evaluation scale across system sizes.
These runs test transfer of the default $n=4$ calibration, not a
scaling law.
With $K_{\max}=8$, truncation bias remains unresolved at
$n \ge 6$, and the default stopping parameters require recalibration
using the closed-form threshold expressions of
Appendix~\ref{sec:theory-stopping-params}.
The Wilson CIs at 10 independent replicates are wide
(e.g., $[0.72, 1.00]$ for FSR$=1.00$), so the larger-system data
identify the required recalibration rather than a full scaling curve.

For $n=4$, each $(p_\phi,\text{rule})$ combination is evaluated over
50 independent replicates.
For the larger-system calibration-transfer tests, each setting is evaluated over
10 independent replicates.
For every replicate we retain the stopping outcome, final $(K,M)$,
effective terminal sample count, estimate $\widehat{\QFI}$, and post
hoc absolute error
$E = |\widehat{\QFI} - \QFI_{\mathrm{ref}}|$.

\paragraph{Decision metrics.}
\label{sec:exp-metrics}

The evaluation metrics are defined on the stopping decision rather
than on the estimator alone.
All three are computed post hoc using $\QFI_{\mathrm{ref}}$.

The \emph{false-stop rate} (FSR) measures how often the
routine declares success while the post hoc
error exceeds tolerance:
\begin{equation}
\label{eq:fsr}
\mathrm{FSR}
=
\frac{\#\{\text{success declaration and } E > \varepsilon\}}
     {\#\{\text{all runs}\}}.
\end{equation}
The \emph{stop rate} (SR) measures how often the routine
declares success:
\begin{equation}
\label{eq:sr}
\mathrm{SR}
=
\frac{\#\{\text{success declaration}\}}
     {\#\{\text{all runs}\}}.
\end{equation}
SR is reported alongside FSR because a zero FSR can have two different
meanings: either the routine made accurate success declarations, or it
made no success declarations and reached the resource limit.
The \emph{stop precision} (SP) measures the conditional
accuracy of success declarations:
\begin{equation}
\label{eq:sp}
\mathrm{SP}
=
\frac{\#\{\text{success declaration and } E \le \varepsilon\}}
     {\#\{\text{success declaration}\}},
\end{equation}
defined when the denominator is nonzero.
Thus FSR is the marginal false-success declaration rate, SR is the
success-declaration rate, and SP is the conditional precision of
success declarations.
When the SP denominator is nonzero, they satisfy the identity
$\mathrm{FSR} = \mathrm{SR}\cdot(1-\mathrm{SP})$, so a
zero observed FSR achieved with a nonzero SR provides evidence
of both safety and usefulness, whereas a zero FSR with zero SR
indicates a zero-success-declaration regime.
Ninety-five percent Wilson confidence intervals are reported
for FSR and SR throughout.

\subsection{False stops arise when sampling convergence precedes
Krylov convergence}
\label{sec:exp-core}

The first result establishes the mechanism that the stopping rule must
control: a narrow sampling interval at the terminal resource pair need
not imply an accurate QFI estimate.
Figure~\ref{fig:width-vs-err} plots the termination-time empirical
interval width $w_M$ against the post hoc absolute error $E$ for all
$n=4$ runs.
Under width-only stopping, many runs fall in the region
$w_M\le\varepsilon$ but $E>\varepsilon$.
These are false stops in the sense of Eq.~\eqref{eq:false-stop-def}:
the empirical sampling-width gate passes, yet the terminal estimate is
outside tolerance.
The post hoc reference value and higher-$K$ trajectories show that the
missing component is Krylov resolution, consistent with
Remark~\ref{rem:false-stop-nec}.
The component-aware rule removes this pattern in the tested grid by
preventing success declarations before the minimum resolution,
sampling, and persistence gates have passed.
This figure is therefore the bridge between the theory of
Section~\ref{sec:stopping} and the aggregate reliability curves below:
it shows why stopping needs a resolution-side condition in addition to
a sampling-width condition.

\begin{figure}[H]
    \centering
    \includegraphics[width=\columnwidth]{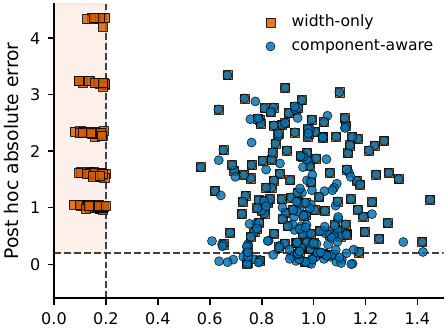}
    \caption{
    \textbf{Narrow empirical intervals do not establish accuracy under
    width-only stopping.}
    Termination-time empirical interval width $w_M$ versus
    post hoc absolute error $E$ for all runs
    ($n=4$, five dephasing levels, 50 independent replicates per noise level).
    Width-only runs form a cluster in
    the upper-left region ($w_M \le \varepsilon$,
    $E > \varepsilon$): these are false stops in the sense of
    Eq.~\eqref{eq:false-stop-def}.
    The empirical gate can pass at the terminal low-$K$ point, but
    higher-$K$ trajectories and the reference value reveal the
    unresolved truncation bias identified in
    Remark~\ref{rem:false-stop-nec}.
    Component-aware stopping removes this upper-left cluster by
    adding eligibility and persistence gates motivated by the
    two-component condition of Theorem~\ref{thm:certified-stop}.
    Dashed lines mark the stopping threshold
    $w_M=\varepsilon$ and the post hoc success threshold
    $E=\varepsilon$.
    }
    \label{fig:width-vs-err}
\end{figure}

\subsection{Aggregate reliability reveals the safety--cost tradeoff
(\texorpdfstring{$n=4$}{n=4})}
\label{sec:exp-aggregate}

The second result shows that the false-stop mechanism is systematic,
not a single representative trajectory.
Figure~\ref{fig:false-stop-noise} reports FSR and SR versus
$p_\phi$ with Wilson confidence intervals, and
Fig.~\ref{fig:median-err-noise} reports the corresponding terminal
absolute errors.
The larger-system runs are not pooled into these curves; they are
reported separately in Appendix~\ref{app:scaling} because they test
calibration transfer rather than the primary $n=4$ claim.

Width-only stopping produces frequent false success declarations.
Its FSR rises from $0.16$ (95\% CI $[0.08, 0.29]$) at
$p_\phi=0$ to $0.68$ (95\% CI $[0.54, 0.79]$) at
$p_\phi=0.24$.
Moreover, every width-only success declaration in the tested grid is
incorrect post hoc
($\mathrm{SP}=0.00$ at all noise levels).
Component-aware stopping has $\mathrm{FSR}=0$ at every tested point,
but this zero must be interpreted together with SR.
Under the default resource limit, all component-aware runs reach the
resource limit without declaring success
($\mathrm{SR}=0.00$, 95\% CI $[0.00,0.07]$ for each noise level).
Thus the powered $n=4$ benchmark supports a focused claim: the component-aware
rule suppresses false success declarations in this fixed-resource
configuration by avoiding premature success labels.

The same aggregate data also show the cost of this reliability gain.
Component-aware runs have median effective terminal sample count $512$ at
every tested dephasing level.
At $p_\phi=0.18$ and $0.24$, width-only runs terminate with median
effective sample counts of $64$ and $32$ shots, respectively, so the
component-aware rule uses up to $16\times$ more terminal samples.
The extra samples change the terminal estimates: median absolute error
falls from $1.553$ to $0.393$ at $p_\phi=0.18$, and from $1.006$ to
$0.211$ at $p_\phi=0.24$.
The apparent efficiency of width-only stopping therefore reflects
early termination at a biased value, not reliable convergence.

\begin{figure}[H]
    \centering
    \includegraphics[width=\columnwidth]{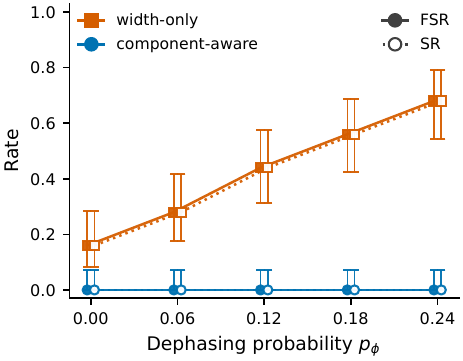}
    \caption{
    \textbf{Component-aware stopping suppresses false success
    declarations under the default resource limit.}
    False-stop rate (FSR, filled markers) and stop rate (SR,
    open markers) versus dephasing probability $p_\phi$
    ($n=4$, 50 independent replicates per point), with 95\% Wilson confidence
    intervals.
    Width-only FSR increases with $p_\phi$ and reaches $0.68$
    at $p_\phi=0.24$, while component-aware FSR remains at zero
    throughout.
    Component-aware SR is also zero at every tested noise level, so
    this fixed-resource grid measures false-success suppression;
    positive success declarations are tested in the
    recalibrated sample-count control of
    Section~\ref{sec:relative-recalibration}.
    }
    \label{fig:false-stop-noise}
\end{figure}

\begin{figure}[H]
    \centering
    \includegraphics[width=\columnwidth]{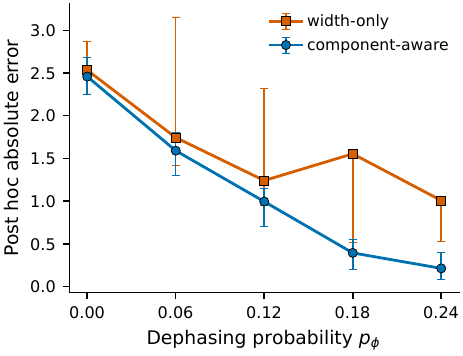}
    \caption{
    \textbf{Width-only stopping's low effective cost reflects premature
    termination at a biased value.}
    Median post hoc absolute error versus $p_\phi$
    ($n=4$, 50 independent replicates per point).
    The two rules diverge at
    $p_\phi \ge 0.18$: component-aware stopping reduces median error
    by up to a factor of $4.8$ relative to width-only stopping despite
    using up to $16\times$ larger effective terminal sample counts at these noise
    levels, indicating that width-only stopping's apparent
    efficiency arises from early termination at a
    truncation-biased estimate, the mechanism of
    Remark~\ref{rem:false-stop-nec}.
    }
    \label{fig:median-err-noise}
\end{figure}

\subsection{Recalibration converts suppression into reliable success}
\label{sec:relative-recalibration}

The fixed-resource benchmark tests the component-aware rule's ability
to avoid premature success labels.
The complementary question is whether the same rule can declare success
once the two resource directions have been calibrated to the target.
We test this in a recalibrated positive control at the low-dephasing point
$p_\phi=0.03$, where
$\QFI_{\mathrm{ref}}=3.7696$ and a true 5\% relative target gives
$\varepsilon=0.1885$.
Exact-density Krylov calibration motivates using full resolution:
at the tested $n=4$ points, $K=12$ still leaves relative truncation
errors from $2.24\%$ to $16.80\%$, whereas $K=16$ removes this
component for the $n=4$ Hilbert space.
We therefore fix $K=16$ and run the sample-count sequence
$M\in\{32768,65536,131072,262144\}$.
The component-aware schedule requires a sampling-width pass to persist
for two consecutive sample-count levels, except at the predeclared terminal
resource limit where a final pass is allowed to declare success.

This recalibrated control produces component-aware success declarations
without observed false stops.
Over 20 independent replicates, the component-aware schedule declares
success in 12 runs
($\mathrm{SR}=0.60$, 95\% CI $[0.39,0.78]$), and every declaration is
correct within the 5\% relative tolerance
($\mathrm{FSR}=0.00$, 95\% CI $[0.00,0.16]$; $\mathrm{SP}=1.00$).
All 20 component-aware terminal estimates are within tolerance; the
eight runs without a success label occur because the empirical bootstrap width remains
above the stopping threshold at the terminal resource limit.
The matched width-only baseline is more willing to stop but less
reliable: it declares success in 19 of 20 runs and makes 15 false
success declarations
($\mathrm{FSR}=0.75$, 95\% CI $[0.53,0.89]$).
The component-aware median relative error is $0.012$ (IQR
$0.009$--$0.017$), while the median final bootstrap width is
$0.175$ (IQR $0.145$--$0.200$), close to
$\varepsilon=0.1885$.

Figure~\ref{fig:relative-recalibration-width-error} shows the decision
geometry behind this control.
All component-aware terminal points lie below the true-error threshold,
while runs without a success label are precisely the cases where the bootstrap width
remains above threshold.
The contrast with the default grid is the central message: zero false
success under insufficient resources is only the safety half of the
story, whereas calibrated Krylov resolution and calibrated sample count
turn the same decision rule into a useful success-declaration
procedure.

\begin{figure}[H]
\includegraphics[width=\columnwidth]{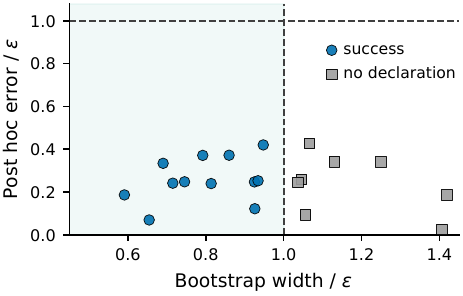}
\caption{
\textbf{Component-aware calibration geometry at true 5\% relative tolerance.}
Each point is a component-aware run from the sample-count sequence,
normalized by
$\varepsilon=0.1885$.
All terminal estimates lie below the true-error threshold
(horizontal dashed line).
The component-aware rule declares success only when the bootstrap width
is also below threshold (vertical dashed line); otherwise it
makes no success declaration at the terminal resource limit.
}
\label{fig:relative-recalibration-width-error}
\end{figure}

\paragraph{High-dephasing calibration frontier.}
\label{sec:calibration-frontier}

The low-dephasing recalibration above uses a stringent 5\% relative
target.
The high-dephasing frontier at $p_\phi=0.24$, where the main
benchmark has the largest width-only FSR, shows the same transition
under a weaker absolute tolerance.
Changing the tolerance and increasing the resource limits moves the
component-aware rule from returning no success label to correct post
hoc success declarations.
Because $\mathcal{F}_{\mathrm{ref}}=1.0561$ at this point, the
absolute tolerances below should not be read as 5\% relative-accuracy
claims.

At $\varepsilon=0.55$, component-aware stopping makes no success
declaration over 20 independent replicates.
At $\varepsilon=0.80$, using $K=12$, $M_{\max}=2048$, and $P=1$, it
declares success in 12 of 20 runs with no observed false stops.
A separate 50-replicate run at $\varepsilon=0.70$ gives 17
component-aware success declarations with no observed false stops,
whereas the matched
width-only baseline declares success in every run and has
FSR$=0.96$.
This frontier reinforces the calibration message: component-aware success
declarations require both Krylov and sampling resources to be matched
to the tolerance, while width-only declarations remain unreliable in
this high-dephasing regime.

\paragraph{Threshold sensitivity.}
\label{sec:threshold-ablation}

The threshold ablation checks whether the eligibility and persistence
gates are doing visible decision work.
We sweep
$K_{\min}^{\mathrm{stop}} \in \{2,4,6\}$,
$M_{\min}^{\mathrm{stop}} \in \{32,128,256\}$, and
$P \in \{1,2,3\}$ at fixed noise
($p_\phi=0.12$, $p_{\mathrm{dep}}=0.03$), $n=4$,
$\varepsilon=0.2$, and resource limits $(K_{\max},M_{\max})=(8,512)$.
Figure~\ref{fig:threshold-ablation-pareto} displays the resulting
false-stop rate as a threshold grid.
Because most configurations in this fixed-resource ablation reach
the resource limit, their median terminal sample count is identical.
The informative variation is therefore the residual false-stop rate
among configurations that are permissive enough to declare success.

The weakest configuration,
$(K_{\min}^{\mathrm{stop}},M_{\min}^{\mathrm{stop}},P)=(2,32,1)$,
exhibits residual false stops.
This is the expected failure mode of an under-calibrated stopping rule:
without enough Krylov resolution, sample count, or persistence, the
empirical tests can pass before the truncation component is controlled.
Most other configurations reach the resource limit, so their zero FSR
is a fixed-resource outcome rather than evidence that all thresholds
are equally useful.
The default $(4,128,2)$ has zero observed false stops and is the
configuration motivated by the threshold calculations in
Appendix~\ref{sec:theory-stopping-params}.
Appendix~\ref{app:threshold-ablation-table} explains why this fixed-resource
ablation is summarized as a threshold grid rather than as a long
enumeration of nearly identical rows.

\begin{figure}[H]
\centering
\includegraphics[width=\columnwidth]{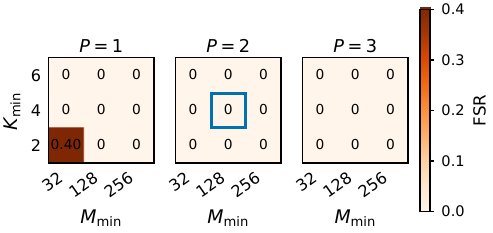}
\caption{
\textbf{Threshold sensitivity under a fixed-resource ablation.}
Each cell is one stopping configuration
$(K_{\min}^{\mathrm{stop}}, M_{\min}^{\mathrm{stop}}, P)$
aggregated over replicates at fixed noise
($p_\phi=0.12$, $p_{\mathrm{dep}}=0.03$),
$\varepsilon=0.2$, resource limits $(K_{\max},M_{\max})=(8,512)$.
Color and cell text give the false-stop rate.
The only configuration with nonzero FSR is the weakest configuration,
$(2,32,1)$.
The default $(4,128,2)$ is outlined in blue and has zero observed
false stops in this ablation.
}
\label{fig:threshold-ablation-pareto}
\end{figure}

\subsection{Krylov convergence sets the calibration scale}
\label{sec:exp-krylov}

The final result explains why the default calibration works in the
powered $n=4$ benchmark but does not automatically transfer to larger
systems.
It checks whether the local Krylov-stability quantity $d_K$ tracks the
truncation side of the error.
Assumption~\ref{ass:krylov-conv} postulates that the truncation bias
$|B_K|$ decays exponentially in $K$.
Figure~\ref{fig:krylov-convergence} compares $d_K$ with the post hoc
truncation bias $|B_K|$ for $n \in \{4,6,8\}$ at $p_\phi=0.12$.

\begin{figure*}[t]
\centering
\includegraphics[width=\textwidth]{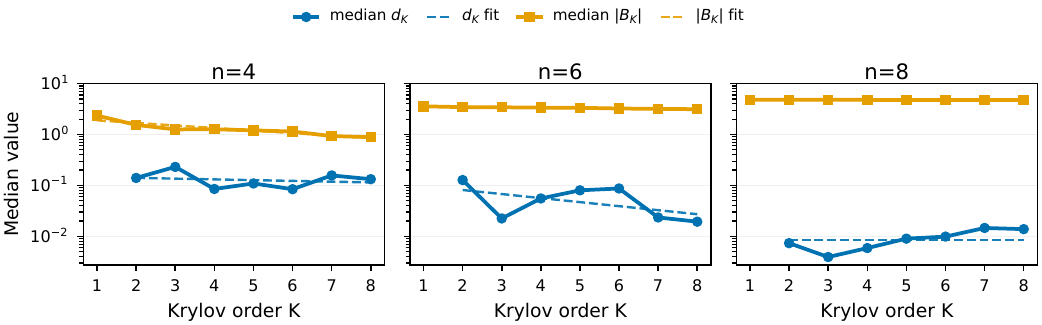}
\caption{
\textbf{Empirical Krylov convergence and observable stability.}
For each system size, markers show median inter-order change
$d_K$ and post hoc truncation bias $|B_K|$ versus Krylov order at
$p_\phi=0.12$; dashed curves show fitted exponential trends.
For $n=4$, both quantities decay with overlapping fitted rates.
For $n\in\{6,8\}$, the bias is nearly flat within
$K\le K_{\max}=8$, explaining why the default stopping rule requires
larger Krylov cutoffs at scale.
The separation between $d_K$ and $|B_K|$ shows that local stability
can be much smaller than the unresolved absolute bias.
}
\label{fig:krylov-convergence}
\end{figure*}

For $n=4$, the bias decays within $K=1,\ldots,8$ with fitted decay
rate $\hat{\mu}(B_K)=0.89$ (95\% CI $[0.86,0.93]$).
The observable inter-order change gives a consistent but noisier rate,
$\hat{\mu}(d_K)=0.97$ (95\% CI $[0.84,1.11]$), supporting its use as a
local Krylov-convergence score in the primary benchmark.
The wider interval for $d_K$ reflects the limited number of full
$K=1,\ldots,8$ trajectories.

For $n=6$ and $n=8$, the same comparison shows where the default
calibration must be strengthened.
The fitted truncation-bias rates are $\hat{\mu}(B_K)=0.98$ and
$1.00$, respectively, so the bias is nearly flat within $K\le8$.
This explains the larger-system calibration-transfer failures in
Appendix~\ref{app:scaling}: when the truncation bias remains large at
$K_{\max}$, a stopping rule calibrated for $n=4$ cannot make a reliable
tolerance claim.
These data motivate larger Krylov cutoffs and recalibration at
$n\ge6$, rather than a scaling conclusion from the default settings.

\paragraph{Summary of the benchmark study.}
Taken together, the benchmark results support a single decision-layer
story.
Width-only stopping fails because it treats sampling convergence as if
it were total-error convergence.
The component-aware rule removes this false-success mode by requiring
resolution-side, sampling-side, and persistence evidence before a
success label is issued.
Under the default $n=4$ resource limit this produces termination
without false success declarations; after recalibrating
$K$ and $M$, it produces correct success declarations at true 5\%
relative tolerance.
The threshold ablation and Krylov-convergence evidence explain the
remaining requirement: reliable success requires thresholds and Krylov
cutoffs matched to the target regime, not a universal fixed setting.

\section{Discussion}
\label{sec:discussion}

\subsection{What this work establishes}

This work separates a QFI estimate from the decision to call that
estimate accurate.
For adaptive Krylov-shadow QFI estimation, this distinction is
essential because the two resources in the estimator control different
objects.
The sample count controls finite-sample spread at the current Krylov
order, while the Krylov order controls whether the fixed-$K$ target is
close to the full QFI\@.
The false stops documented here occur when the sampling-side evidence
looks favorable before the Krylov component has been resolved.

The theoretical contribution is therefore a stopping contract, not a
new fixed-resource QFI formula.
Theorem~\ref{thm:certified-stop} states the component condition: a
success declaration is justified when calibrated truncation and
sampling radii are both controlled at the tolerance scale.
Proposition~\ref{prop:finite-horizon-heldout} gives the corresponding
adaptive finite-horizon interface by separating exploratory candidate
selection from an independent held-out confirmation batch.
This distinction also fixes the language of the paper.
The same-sample component-aware rule provides observed reliability
under the tested calibrated resources; the word certificate is reserved
for the held-out rule with calibrated component radii.

The numerical evidence shows how this contract behaves in the benchmark.
In the powered $n=4$ fixed-resource grid, width-only stopping makes
frequent false success declarations, with FSR between $0.16$ and
$0.68$.
The component-aware empirical rule makes no success declaration under
that default resource limit and therefore has no observed false success
declarations.
This is not a weakness of the interface: in a regime where the
evidence does not support the tolerance claim, termination without a
success label is the intended output.
The recalibrated 5\% relative-tolerance control supplies the
positive-control case.
After increasing Krylov resolution and sample count, the
component-aware rule declares success in 12 of 20 runs with no observed
false stops,
whereas the matched width-only baseline declares more often but makes
15 false success declarations.

The reliability gain has an explicit resource cost.
Component-aware runs use larger terminal sample counts in the main
benchmark, and in some regimes they end without a success label.
This tradeoff is part of the design rather than an implementation
accident.
The interface is meant to distinguish three outcomes that a point
estimate alone cannot: reliable success, false success, and
termination at the resource limit without a success label.
The classical overhead is modest in the present implementation: each
bootstrap evaluation uses $B$ resamples, giving $O(BM)$ classical
post-processing at sample count $M$.
The sample-count summaries count randomized-measurement shots, not
bootstrap resamples.

\subsection{Separation of concerns in adaptive estimation}

The estimator and the stopping rule should be treated as different
modules.
The fixed-resource estimator maps a chosen pair $(K,M)$ to a QFI
estimate.
The stopping rule decides whether that estimate should be returned
with a success label.
The false stops in this paper arise from the second module: the
estimate can still move with larger $K$, but the decision layer has
already accepted a low-$K$ value.

This separation makes the interface reusable.
A stricter decision rule can be placed around the same fixed-resource
estimator, and the terminal summary identifies the binding obstruction:
Krylov resolution, sampling width, or the resource limit.
The threshold expressions for
$K_{\min}^{\mathrm{stop}}$ and
$M_{\min}^{\mathrm{stop}}$ in Appendix~\ref{sec:theory-stopping-params} should
therefore be read as recalibration formulas for new regimes; their
numerical values depend on the target family and tolerance.
The larger-system calibration-transfer tests illustrate this point directly:
when $K_{\max}=8$ is too small for the truncation bias to decay, the
$n=4$ stopping parameters do not transfer automatically.

\subsection{Open directions}

The main practical requirement is larger-system calibration.
At $n=6$ and $n=8$, the default $K_{\max}=8$ cutoff leaves the
truncation bias nearly flat over the tested range.
Stronger larger-system claims will require larger Krylov cutoffs,
better seed or generator choices, or problem-specific calibration of
the truncation radius.
The present paper uses these runs to identify the next calibration
target for larger systems.

The main statistical requirement is sequential calibration.
The held-out rule already handles data-dependent candidate selection
through independent confirmation data and finite-horizon alpha
spending.
Its remaining external input is a calibrated Krylov truncation radius.
When such component radii are available,
Theorem~\ref{thm:certified-stop} does not require the reference value
$\QFI_{\mathrm{ref}}$ during operation; the reference value is used
only for post hoc benchmark labeling.
A natural next step is to replace or tighten the calibration table by
time-uniform component radii for the observable checks $(d_K,w_M)$.
Confidence sequences~\cite{howard2021anytime} provide one route
because their probability bounds hold uniformly along adaptive
trajectories.

\section{Conclusion}
\label{sec:conclusion}

Adaptive QFI estimation needs a reliable success label as well as an
efficient numerical estimator.
This paper identifies a concrete failure of that label in
Krylov-shadow QFI estimation: empirical sampling width can become small
around a biased low-order Krylov estimate.
AKS-QFI addresses the failure at the decision layer by requiring
resolution-side and sampling-side evidence before reporting success.

The main benchmark identifies the reliability gain targeted by AKS-QFI.
For $n=4$ noisy mixed states, the width-only rule makes frequent false
success declarations, while the component-aware empirical rule returns
no success label under the default resource limit and has no observed false
success declarations.
After recalibrating Krylov resolution and sample count to a true
5\% relative target, the same component-aware logic produces correct
success declarations in a recalibrated positive control.
Together, these results show that the success label becomes meaningful
when Krylov resolution and sampling precision are calibrated together.

The held-out construction gives the route from observed reliability to
a finite-horizon certificate.
When calibrated component radii and an independent confirmation batch
are available, the final success label can be attached to the
confirmation estimate rather than to the exploratory estimate.
The remaining challenge is calibration at scale: larger systems require
larger or better-tuned Krylov cutoffs, and a fully powered held-out
large-system study requires independent confirmation batches and
precomputed truncation radii for each problem family.
Thus the present contribution is a certifiable stopping interface and
empirical evidence that this interface is needed for scalable QFI
workflows.

\section*{Data Availability}

The code, processed benchmark data, plotting scripts, and figure
reproduction package needed to reproduce the numerical results in this
article are available in the project repository accompanying this
manuscript.
A permanent archival version of the repository will be deposited upon
publication.
No quantum-hardware data are required to reproduce the reported
benchmarks.

\appendix

\section{Algorithmic Details}
\label{app:impl}

This appendix records the decision logic needed to reproduce the
stopping outcomes reported in the main text.
Algorithm~\ref{alg:aks-qfi} should be read as the adaptive interface
around a fixed-resource Krylov-shadow estimator.
The algorithm records both the returned estimate and the stopping
label, because the benchmark evaluates the reliability of that label.
The benchmark configuration is specified separately in
Appendix~\ref{app:bench}.

\refstepcounter{algorithm}
\label{alg:aks-qfi}
\smallskip
\noindent\textbf{Algorithm~\thealgorithm. AKS-QFI: Adaptive Krylov-Shadow QFI
Estimation}
\par\smallskip
\noindent\rule{\columnwidth}{0.4pt}
\par\smallskip
{\footnotesize
\begin{algorithmic}[1]
\Require tolerance $\varepsilon>0$; resource limits $K_{\max}$,
         $M_{\max}$; initial order $K_0$, initial sample count
         $M_0$; stopping parameters
         $K_{\min}^{\mathrm{stop}}$,
         $M_{\min}^{\mathrm{stop}}$, patience $P$
\Ensure point estimate $\widehat{\QFI}$; stopping summary
\State $K \leftarrow K_0$,\; $M \leftarrow M_0$,\;
       $p \leftarrow 0$ \Comment{patience counter}
\Repeat
  \State Compute $\widehat{\QFI}_{K,M}$ via the
         Krylov-shadow pipeline~\cite{zhang2025krylovshadow}
  \If{$K>1$}
    \State Compute $\widehat{\QFI}_{K-1,M}$ without bootstrap
    \State $\widetilde I^{(\mathrm{trunc})}_{K,M} \leftarrow
           |\widehat{\QFI}_{K,M} - \widehat{\QFI}_{K-1,M}|$
           \hfill (Eq.~\eqref{eq:I-trunc})
  \Else
    \State Set $\widetilde I^{(\mathrm{trunc})}_{K,M}\leftarrow \infty$
           \Comment{force one Krylov increase}
  \EndIf
  \State Set $\widetilde I^{(\mathrm{stat})}_{K,M}\leftarrow
         W^{(\mathrm{stat})}_{K,M}=w_M$ via
         Eq.~\eqref{eq:I-stat}
  \State $\widetilde I_{K,M} \leftarrow
         \max\{\widetilde I^{(\mathrm{trunc})}_{K,M},
         \widetilde I^{(\mathrm{stat})}_{K,M}\}$
  \If{$K \ge K_{\min}^{\mathrm{stop}}$ \textbf{and}
      $M \ge M_{\min}^{\mathrm{stop}}$
      \textbf{and}
      $\bigl(\widetilde I^{(\mathrm{trunc})}_{K,M}\le\varepsilon$
      \textbf{or} $K=K_{\max}\bigr)$
      \textbf{and}
      $\widetilde I^{(\mathrm{stat})}_{K,M}\le\varepsilon$}
    \State $p \leftarrow p + 1$
  \Else
    \State $p \leftarrow 0$
  \EndIf
  \If{$p \ge P$}
    \If{held-out certificate mode is enabled}
      \State Draw an independent confirmation batch at the candidate
             pair $(K,M_{\mathrm{conf}})$
      \State Compute
             $R_{\mathrm{tot}}
             \leftarrow R^{(\mathrm{trunc})}_{K}
             + R^{(\mathrm{stat})}_{\mathrm{conf}}(\delta_j)$
      \If{$R_{\mathrm{tot}}\le\varepsilon$}
        \State \textbf{return} the held-out confirmation estimate
               with a held-out certified success declaration
      \Else
        \State $p \leftarrow 0$ and continue increasing $K$ or $M$
      \EndIf
    \Else
      \State \textbf{return} $\widehat{\QFI}_{K,M}$ with
             a success declaration
    \EndIf
  \EndIf
  \If{$K \ge K_{\max}$ \textbf{and} $M \ge M_{\max}$}
    \State \textbf{return} $\widehat{\QFI}_{K,M}$ with
           a resource-limit termination
  \EndIf
  \If{$\widetilde I^{(\mathrm{trunc})}_{K,M} > \varepsilon$
      \textbf{and} $K < K_{\max}$}
    \State $K \leftarrow K + 1$
  \ElsIf{$\widetilde I^{(\mathrm{stat})}_{K,M} > \varepsilon$
      \textbf{and} $M < M_{\max}$}
    \State $M \leftarrow 2M$
  \ElsIf{$K < K_{\max}$}
    \State $K \leftarrow K + 1$
      \Comment{final Krylov pass before the resource limit}
  \ElsIf{$M < M_{\max}$}
    \State $M \leftarrow 2M$
      \Comment{final sampling pass before the resource limit}
  \Else
    \State \textbf{return} $\widehat{\QFI}_{K,M}$ with
           a resource-limit termination
  \EndIf
\Until{a return condition is met}
\end{algorithmic}
}
\par\smallskip
\noindent\rule{\columnwidth}{0.4pt}
\par\smallskip

\subsection{Benchmark summaries}
The run records separate trajectory-level data from aggregate
reliability summaries.
The main text reports the aggregate quantities tied to the stopping
claim: stop rate, false-stop rate, stop precision, terminal sample count,
and terminal error.
Individual trajectories are used only to illustrate representative
failure modes and resource-allocation behavior.

Each terminal run is classified as either a success declaration or a
resource-limit termination.
For component-aware runs, the record also stores whether the
minimum-$K$ gate, minimum-$M$ gate, component gates, and persistence
requirement passed.
These fields explain the decision path of a run; they are not
additional scientific claims beyond the stopping outcome itself.

\subsection{Seed vector selection}
\label{app:seed-vector}

The Krylov subspace $\mathcal{K}_K(G,v_0)$ is constructed by applying
the Hermitian generator $G$ iteratively to an initial seed vector
$v_0$.
In the reported benchmark, $v_0$ is the normalized dominant
eigenvector of the noisy density matrix.
This deterministic choice removes seed-selection randomness from the
primary stopping study.
It is not, however, guaranteed to align with the most QFI-sensitive
directions.
Poor alignment can increase the effective truncation constant
$C_{\mathrm{trunc}}$ and delay Krylov convergence, as discussed in
Appendix~\ref{sec:theory-decomp}.
The $n=4$ benchmark shows acceptable decay for this seed choice, while
the $n=6,8$ calibration-transfer tests show that larger systems may require a
different seed strategy or larger Krylov cutoffs.
A systematic seed-vector study is therefore a future calibration task
rather than a claim of the present benchmark.

\subsection{Reproducibility package}
\label{app:code-availability}

The repository contains the estimator implementation, benchmark
configurations, processed run summaries, plotting scripts, and figure
replot package used for the reported statistics and figures.
Each numerical figure is generated from processed CSV summaries or
stored run artifacts.
No quantum-hardware data are required to reproduce the reported
benchmarks.

\section{Benchmark Configuration}
\label{app:bench}

This appendix specifies the benchmark family and resource settings used
for the stopping-reliability study.
The primary grid is the $n=4$ noisy mixed-state benchmark reported in
the main text.
The dephasing probability is varied over
\[
p_\phi \in \{0,\,0.06,\,0.12,\,0.18,\,0.24\},
\]
while the depolarizing probability is fixed at $p_{\mathrm{dep}}=0.03$.
For $n=4$, we use 50 independent replicates for each
$(p_\phi,\,\text{rule})$ combination, giving
$50 \times 5 = 250$ terminal outcomes per rule and 500 terminal
outcomes in total
for the $n=4$ benchmark grid.
The larger-system runs at $n \in \{6,8\}$ use 10 independent
replicates per setting and are reported separately as calibration
transfer tests in Appendix~\ref{app:scaling}.

The main reliability study uses the fixed absolute tolerance
$\varepsilon=0.2$ described in Section~\ref{sec:exp-setup}.
The larger-system calibration-transfer tests in
Appendix~\ref{app:scaling} use the relative tolerance
$\varepsilon_{\mathrm{rel}}=0.05$ to normalize the evaluation scale
across $n$.
The main reliability study compares two stopping policies built on the
same estimator pipeline: a width-only empirical rule and a
component-aware rule with eligibility and persistence constraints.
For the component-aware configuration, the default constraints are
\[
K_{\min}^{\mathrm{stop}}=4,\qquad
M_{\min}^{\mathrm{stop}}=128,\qquad
P=2.
\]
These defaults are calibrated to the primary $n=4$ benchmark; larger
systems require recalibration.

\paragraph{Noisy mixed-state family.}
The benchmark family is constructed as follows.
We first prepare an entangled pure state from
\begin{align*}
|\psi_\alpha\rangle
&= \exp(-i\alpha H_{\mathrm{ent}})|+\rangle^{\otimes n}, \\
H_{\mathrm{ent}}
&= \sum_{j=1}^{n-1} Z_j Z_{j+1}
  + 0.35\sum_{j=1}^{n} X_j,
\end{align*}
with $\alpha=0.25$ in the reported benchmark.
The noisy base state is
\begin{equation}
\label{eq:noisy-mixed}
\rho
= \mathcal{E}_{\mathrm{dep}}(p_{\mathrm{dep}})
  \circ \mathcal{E}_{\mathrm{deph}}(p_\phi)
  \bigl(|\psi_\alpha\rangle\!\langle\psi_\alpha|\bigr),
\end{equation}
where $\mathcal{E}_{\mathrm{deph}}(p_\phi)$ applies independent
single-qubit dephasing channels
$\rho \mapsto (1-p_\phi)\rho + p_\phi Z_j \rho Z_j$
to each qubit $j$, and
$\mathcal{E}_{\mathrm{dep}}(p_{\mathrm{dep}})$ applies a global
depolarizing channel
$\rho \mapsto (1-p_{\mathrm{dep}})\rho
  + p_{\mathrm{dep}} I/2^n$.
The parameter is then encoded by the unitary family
\begin{align*}
\rho_\theta &= e^{-iG\theta}\rho\,e^{iG\theta},
\\
G &= \frac{1}{2}\sum_{j=1}^{n} Z_j
  + 0.08\sum_{j=1}^{n-1} X_j X_{j+1}.
\end{align*}
The QFI is evaluated at $\theta=0$; the reference value
$\QFI_{\mathrm{ref}}$ is computed via the exact spectral
decomposition formula~\eqref{eq:qfi-spectral} applied to the
analytically tractable $2^n \times 2^n$ density matrix.

The family isolates stopping failures in a setting where an exact
reference remains available for evaluation.
It is a controlled benchmark rather than a survey of problem families,
hardware noise models, or hardware-scaled regimes.

\section{Status of Assumption-Explicit Guarantees}
\label{app:theory-status}

This appendix clarifies the relation between the assumption-calibrated
statements in the theory and the empirical checks used in the
reported stopping rule.

Theorem~\ref{thm:certified-stop} gives a sufficient condition for
reliable stopping when calibrated component bounds control the
truncation and sampling errors at the terminal resource pair.
The empirical rule uses the inter-order change $d_K$ and bootstrap
width $w_M$ as observable scores tied to those components.
The bootstrap width captures finite-sample spread at the current
$K$, while the inter-order change is a local measure of Krylov
stability.
Neither score alone is a distribution-free bound on total QFI
error.
This distinction is central to the paper: width-only stopping fails
because the sampling width can be small around a biased low-$K$
estimate while the truncation-side error remains large.

The remaining theoretical step toward a fully calibrated online
routine is to replace these terminal empirical checks with
time-uniform, sequentially valid component radii along the adaptive
trajectory.
The present numerical study instead evaluates the decision layer
post hoc, using $\QFI_{\mathrm{ref}}$ only to label terminal stopping
decisions as accurate or false in the benchmark.

\section{Threshold-Ablation Summary}
\label{app:threshold-ablation-table}

The threshold-ablation sweep covers all 27 settings of
$(K_{\min}^{\mathrm{stop}},M_{\min}^{\mathrm{stop}},P)$ in the
grid shown in Fig.~\ref{fig:threshold-ablation-pareto}.
We report a threshold-grid summary rather than a long table because
the fixed-resource design makes most settings lead to the same
qualitative outcome: resource-limit termination without a success
declaration.

\paragraph{Why 26 of 27 configurations show FSR\,$=\!0.00$.}
The ablation is evaluated at a single fixed noise point
($p_\phi=0.12$, $p_{\mathrm{dep}}=0.03$) under the fixed
resource limit $(K_{\max},M_{\max})=(8,512)$.
Under this resource limit, the dominant outcome across almost
all configurations is resource-limit termination rather
than a success declaration: the truncation bias at this noise
level is large enough that most runs exhaust the resource limit
before the indicator falls below $\varepsilon=0.2$.
When no run in a configuration declares success,
the FSR numerator is identically zero by definition
(Eq.~\eqref{eq:fsr}), regardless of the stopping settings.
This makes the grid nearly uniform in FSR and limits its
discriminative power as an ablation of all stopping parameters.
The single exception, $(K_{\min},M_{\min},P)=(2,32,1)$ with
FSR$=0.40$, occurs because very low thresholds allow early
success declarations before the resource limit is engaged.
A more discriminative ablation would use a reduced resource limit,
for example $M_{\max}=128$, to force more runs through the
success-declaration branch.

\section{Additional Benchmark Figures}
\label{app:supplementary-checks}

This appendix collects supplementary evidence for the $n=4$
benchmark.
Unless stated otherwise, the fixed-noise examples use
$p_\phi=0.12$ and $p_{\mathrm{dep}}=0.03$.
The figures do not introduce new claims beyond the main-text
reliability results.
They show the same decision-layer behavior from three additional
perspectives: matched trajectories, effective terminal sample count, and
stop precision.

\paragraph{Adaptive trajectory.}
Figure~\ref{fig:trace-example} illustrates how the false stop appears
along a single matched trajectory.
It shows the stepwise evolution of $(K_j, M_j)$, the point estimate
$\widehat{\QFI}_{K_j,M_j}$, and the composite indicator
$\widetilde I_{K_j,M_j}$ for a representative width-only run and its
matched component-aware run.
Under width-only stopping, the trajectory terminates at
$(K,M)=(2,32)$ in iteration~3: the bootstrap width has
fallen below $\varepsilon$ at low Krylov order, so the
indicator satisfies the stopping criterion before any
truncation-bias reduction has occurred.
The estimate at termination ($\widehat{\QFI}\approx 0.029$)
is far below the reference ($\QFI_{\mathrm{ref}}\approx 2.35$),
illustrating the false-stop mechanism of
Remark~\ref{rem:false-stop-nec}.
Under component-aware stopping, the same trajectory is extended.
The component-aware rule's $K_{\min}=4$ and $M_{\min}=128$ thresholds
prevent a success declaration at low resources, and the run terminates at
$(K,M)=(8,512)$ by reaching the resource limit with a substantially
smaller error.

\begin{figure}[H]
\centering
\includegraphics[width=\columnwidth]{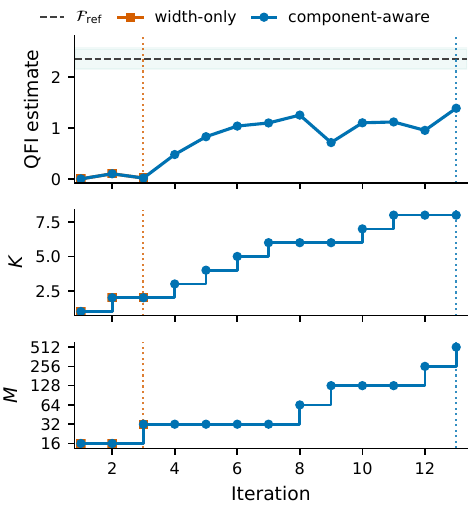}
\caption{
\textbf{Representative adaptive trajectories under width-only and
component-aware stopping.}
Adaptive trajectories for matched width-only and component-aware runs at the
same conditions
($p_\phi=0.12$, $p_{\mathrm{dep}}=0.03$, matched draw;
$\QFI_{\mathrm{ref}}=2.353$, $\varepsilon=0.2$).
This is the same representative instance discussed in
Appendix~\ref{sec:failure_example}.
Top: QFI estimate, with the shaded band marking
$|\widehat{\QFI}-\QFI_{\mathrm{ref}}|\le\varepsilon$.
Middle and bottom: Krylov order $K$ and sample count $M$.
Width-only stopping halts at iteration~3 with $(K,M)=(2,32)$
(a false success declaration), while component-aware stopping
continues to $(K,M)=(8,512)$ and reduces the terminal error.
Dotted vertical lines mark the two stopping iterations.
}
\label{fig:trace-example}
\end{figure}

\paragraph{Effective sample-count distribution at $p_\phi=0.12$.}
The sample-count summaries separate apparent efficiency from reliable
stopping.
At $p_\phi=0.12$, width-only stopping has median effective terminal
sample count $512$ with IQR $[32,512]$.
This wide interval reflects two populations: early success
declarations, many of which are false stops concentrated at
$M\le64$, and resource-limit terminations at $M=512$.
Component-aware stopping has median effective terminal sample count $512$
with IQR $[512,512]$, because all 50 component-aware runs reach the
resource limit.
Thus the component-aware distribution records termination without a
success label under the default resource limit, whereas the low-sample
width-only tail records premature stopping.

\begin{figure}[H]
    \centering
    \includegraphics[width=\columnwidth]{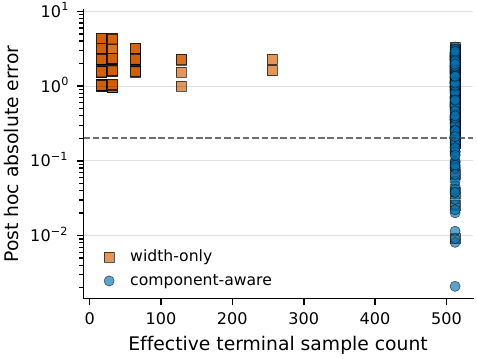}
    \caption{
\textbf{Early termination under width-only stopping can appear
efficient while remaining inaccurate.}
Post hoc absolute error versus effective terminal sample count for the
benchmark (logarithmic vertical axis).
Width-only runs that terminate at low effective sample counts frequently
exhibit large post hoc errors, whereas component-aware runs concentrate near
the resource limit and avoid false success declarations on the tested
benchmark grid.
The plot reports the reliability--cost tradeoff induced by the two
stopping rules.
}
    \label{fig:err-vs-calls}
\end{figure}

\begin{figure}[H]
    \centering
    \includegraphics[width=\columnwidth]{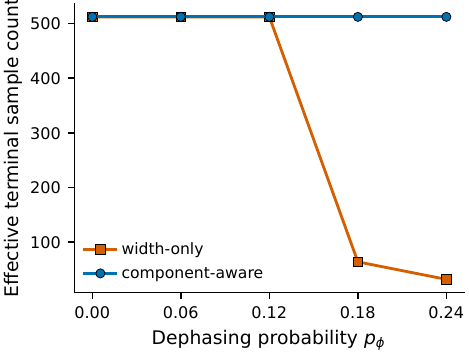}
    \caption{
\textbf{Component-aware stopping uses more terminal samples under the
current resource limits.}
Median effective terminal sample count versus dephasing probability $p$ in
the $n=4$ benchmark (50 independent replicates per point).
In the tested configuration, component-aware stopping consistently reaches
the resource limit (median effective sample count 512), while width-only
stopping often terminates much earlier, especially at higher
dephasing.
Together with Figs.~\ref{fig:false-stop-noise}
and~\ref{fig:median-err-noise}, this indicates that much of the
apparent cost saving of width-only stopping arises from premature
termination.
}
    \label{fig:median-calls-noise}
\end{figure}

\begin{figure}[H]
    \centering
    \includegraphics[width=\columnwidth]{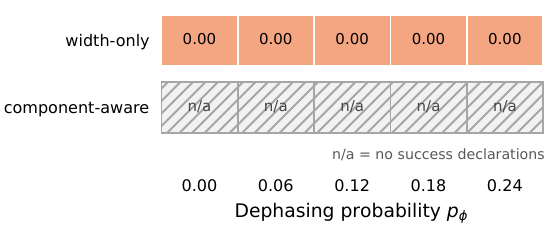}
    \caption{
\textbf{Stop-precision summary for success declarations.}
Each cell reports the post hoc stop precision for a rule and
dephasing probability in the $n=4$ benchmark
(50 independent replicates per point),
defined as the fraction of success declarations satisfying
$E\le \eps$.
Width-only declarations have precision zero at all tested dephasing
levels.
Component-aware entries are marked n/a because component-aware runs
make no success declarations under the current resource limit and
default stopping rule.
The corresponding FSR and SR panels report this resource-limit
outcome.
}
    \label{fig:precision-noise}
\end{figure}

\section{Representative false-stop instance}
\label{sec:failure_example}

This appendix gives a concrete low-resource termination instance that
illustrates the false-stop pathology discussed in
Section~\ref{sec:stopping}.
It is a worked example rather than an additional statistical
claim: the aggregate reliability statements are those reported in
Section~\ref{sec:exp-aggregate}.
The instance shows how a narrow empirical uncertainty interval can
coexist with a large post hoc error when the run is under-resolved in
Krylov order.
All quantities involving $\QFI_{\mathrm{ref}}$ are used only for
\emph{post hoc} evaluation and are unavailable in deployment.

\subsection{Setup and stopping outcome}
We consider the noisy mixed-state benchmark family described in
Appendix~\ref{app:bench} with target tolerance $\eps=0.2$.
Under the width-only stopping rule
(Section~\ref{sec:stopping}), a representative run terminates at
the small resource pair $(K,M)=(2,32)$ by declaring success.
At termination, the empirical uncertainty interval width is
approximately $0.153$, below the tolerance.
However, the post hoc absolute error relative to the benchmark
reference is approximately $2.323$, far above $\eps$.
This mismatch is not a near-threshold fluctuation.
It is the decision-layer failure formalized in
Section~\ref{sec:stopping}: the stopping claim is triggered at low
resources before Krylov truncation has been adequately controlled.

\subsection{Interpretation}
This instance highlights the core mechanism behind false stops in the
reliability study.
At low $(K,M)$, the empirical interval can shrink and the
adjacent-order change can appear locally stable before the Krylov
resolution is sufficient.
In noisy mixed-state regimes, the point estimate may therefore remain
substantially biased while still appearing ``precise'' according to
the empirical interval width and the local stability measure.
Component-aware stopping (Section~\ref{sec:stopping}) suppresses this
mode by requiring eligibility thresholds and persistence before
permitting a success claim.
The example therefore isolates the decision-layer issue addressed by
the component-aware rule, without asserting that the default stopping
parameters transfer unchanged to every problem family or system size.

\section{Theory of the Stopping Interface}
\label{app:theory}

This appendix develops the theoretical foundations and assumption
boundaries for the stopping interface of AKS-QFI\@.
Section~\ref{sec:theory-decomp} establishes the error
decomposition and derives the finite-sample bound motivating
Theorem~\ref{thm:certified-stop}.
Section~\ref{sec:theory-bootstrap} characterizes bootstrap
coverage under adaptive stopping.
Section~\ref{sec:theory-stopping} provides the proof of
Theorem~\ref{thm:certified-stop} and the derivation
supporting Remark~\ref{rem:false-stop-nec}.
Section~\ref{sec:theory-stopping-params} derives closed-form expressions
for the stopping parameters introduced in
Section~\ref{sec:stopping}.

\subsection{Error Decomposition and Finite-Sample Bound}
\label{sec:theory-decomp}

\paragraph{Setup.}
Let $\rho_\theta$ be a density operator on a Hilbert space of
dimension $d$, and let $\QFI(\rho_\theta)$ denote the quantum
Fisher information at $\theta$.
Fix a Krylov order $K$ and sample count $M$.
Denote the Krylov-shadow point estimate by
$\widehat{\QFI}_{K,M}$, the Krylov-truncated population value
(the infinite-sample limit of the estimator at order $K$) by
$\QFI_K$, and the target by
$\QFI \coloneqq \QFI(\rho_\theta)$.

\paragraph{Decomposition.}
The total estimation error decomposes as
\begin{equation}
\label{eq:decomp}
\widehat{\QFI}_{K,M} - \QFI
=
\underbrace{\bigl(\widehat{\QFI}_{K,M} - \QFI_K\bigr)%
}_{\text{sampling error } S_{K,M}}
+
\underbrace{\bigl(\QFI_K - \QFI\bigr)%
}_{\text{truncation bias } B_K}.
\end{equation}
The bias $B_K$ is deterministic and controlled solely by $K$,
while $S_{K,M}$ is the sampling-side finite-$M$ error at fixed
$K$.
For a nonlinear plug-in estimator, $S_{K,M}$ can include both
random fluctuation and finite-$M$ plug-in bias.

\begin{assumption}[Krylov convergence]
\label{ass:krylov-conv}
The truncation bias satisfies
$|B_K| \le C_{\mathrm{trunc}}\,\mu^K$
for constants $C_{\mathrm{trunc}}>0$ and $\mu\in(0,1)$ that
depend on the effective spectral structure of the projected QFI
problem and on the generator used in the Krylov construction, but
not on $M$.
\end{assumption}

\begin{assumption}[Sampling-side concentration]
\label{ass:subgauss}
For each fixed $K$, write
$S_{K,M}=\widehat{\QFI}_{K,M}-\QFI_K$.
There exist $\beta_{K,M}\ge0$ and $\sigma_K>0$ such that
\[
\bigl|\mathbb{E}S_{K,M}\bigr|\le \beta_{K,M},
\]
and the centered fluctuation
$S_{K,M}-\mathbb{E}S_{K,M}$ is sub-Gaussian with proxy variance
$\sigma_K^2/M$.
\end{assumption}

Assumption~\ref{ass:krylov-conv} is expected to hold when the
Krylov basis is generated by a sufficiently expressive generator
and the projected QFI problem has a separated dominant subspace;
it can be monitored empirically via the inter-order change $d_K$.
For a linear unbiased shadow observable, $\beta_{K,M}=0$.
For the nonlinear plug-in estimator used in the benchmark,
$\beta_{K,M}$ is the finite-$M$ plug-in bias and must either be
small at the operating sample count or included in the sampling-side
calibration.
The fluctuation part of Assumption~\ref{ass:subgauss} is implied
by bounded per-shot shadow observables together with independent
randomized-measurement shots; the bound is discussed below.

\paragraph{Conditions under which Assumption~\ref{ass:krylov-conv}
may fail.}
The exponential-decay condition $|B_K| \le C_{\mathrm{trunc}}\mu^K$
with $\mu \in (0,1)$ relies on the Krylov basis effectively
approximating the dominant QFI-sensitive subspace.
Three physically relevant cases are important.
\emph{(i) Many-body localization and flat spectra.}
In disordered or many-body-localized spin systems, the relevant
QFI directions can have little spectral separation.
In that regime $\mu \approx 1$ and the bias decays negligibly
within any practical Krylov cutoff, precisely the phenomenon
observed at $n \in \{6,8\}$ (Section~\ref{sec:exp-krylov}).
\emph{(ii) Near-critical points.}
Close to a quantum phase transition, the relevant spectral gap can
close and the correlation length diverges; the effective
$\mu$ approaches $1$ from below, so $K_{\min}^{\mathrm{stop}}$
diverges accordingly.
\emph{(iii) Seed-vector alignment.}
If the initial Krylov seed vector $v_0$ has small overlap with
the dominant QFI-sensitive directions, early Krylov iterates may fail
to capture the high-weight subspace, effectively increasing
$C_{\mathrm{trunc}}$ and delaying convergence.
In all three cases, Assumption~\ref{ass:krylov-conv} may fail over the
tested range, or its useful decay regime may lie beyond the current
Krylov cutoff.
The inter-order change $d_K$ provides local evidence of this behavior.
If $d_K$ remains large near $K_{\max}$, the calibrated criterion would
require a larger Krylov cutoff or no held-out certified success
declaration.
The same-sample empirical rule can still mark the Krylov-side gate as
passed at $K_{\max}$ under a finite-resource convention, so such
decisions must be interpreted as finite-resource outcomes and audited
post hoc.
The theoretical threshold $K_{\min}^{\mathrm{stop}}(\varepsilon)$
in Eq.~\eqref{eq:kmin-theory} diverges as $\mu\to1$, signaling
that the resource limit must be increased.

\paragraph{Verifiability of Assumption~\ref{ass:subgauss}.}
The sub-Gaussian condition on the sampling error $S_{K,M}$
can be justified when the per-shot Krylov-shadow observable is
bounded, or when it is clipped to a bounded range.
If each centered per-shot contribution for a linearized estimator lies
in an interval of length $\Delta_K$, Hoeffding's lemma gives a
sub-Gaussian proxy variance $\Delta_K^2/(4M)$ for the sample mean.
For a concrete shadow estimator, $\Delta_K$ is controlled by the
measurement ensemble, the inverse shadow channel, and the norm of
the linearized projected QFI functional.
For nonlinear plug-in estimators, this argument applies after a local
linearization and a separate bound on the plug-in bias term
$\beta_{K,M}$.
During adaptive operation this bound is usually loose.
The bootstrap width $w_M$ serves as an empirical estimate of the
sampling scale, with coverage properties characterized in
Section~\ref{sec:theory-bootstrap}.
If the assumptions motivating Theorem~\ref{thm:certified-stop} fail
simultaneously, for example $\mu \approx 1$ and the bootstrap
underestimates the true variance, the practical rule can make a
premature success declaration at a resource level where neither
error source is adequately controlled.
The terminal checks identify the values that led to the decision, but
detecting this failure generally requires held-out calibration, a
higher-$K$ comparison, or a benchmark reference value.

\begin{proposition}[Finite-sample error bound]
\label{prop:error-bound}
Under Assumptions~\ref{ass:krylov-conv} and~\ref{ass:subgauss},
for any $\delta\in(0,1)$,
\begin{equation}
\label{eq:error-bound}
\begin{aligned}
\Pr\!\biggl[
  \bigl|\widehat{\QFI}_{K,M} - \QFI\bigr|
  &>
  C_{\mathrm{trunc}}\,\mu^K+\beta_{K,M}\\
  &\quad
  + \sigma_K\sqrt{\frac{2\log(2/\delta)}{M}}
\biggr]\le \delta.
\end{aligned}
\end{equation}
\end{proposition}

\begin{proof}
By the triangle inequality applied to~\eqref{eq:decomp},
\begin{equation*}
\bigl|\widehat{\QFI}_{K,M}-\QFI\bigr|
\le |B_K| + |S_{K,M}|.
\end{equation*}
Since $|B_K| \le C_{\mathrm{trunc}}\mu^K$ holds
deterministically by Assumption~\ref{ass:krylov-conv}, it
suffices to bound $|S_{K,M}|$.
By Assumption~\ref{ass:subgauss},
$|S_{K,M}|\le \beta_{K,M}
+|S_{K,M}-\mathbb{E}S_{K,M}|$.
By the sub-Gaussian tail bound
(Assumption~\ref{ass:subgauss}),
\begin{equation*}
\Pr\!\left[
  |S_{K,M}-\mathbb{E}S_{K,M}|>
  \sigma_K\sqrt{\frac{2\log(2/\delta)}{M}}
\right]\le\delta,
\end{equation*}
and the result follows.
\end{proof}

\paragraph{Relationship to observable checks.}
Proposition~\ref{prop:error-bound} motivates a two-term
calibrated error envelope for AKS-QFI\@.
The theoretical envelope components are
\begin{equation}
\label{eq:indicator-theory}
\begin{aligned}
I^{(\mathrm{trunc})}_{K}
&\coloneqq C_{\mathrm{trunc}}\,\mu^K,\\
I^{(\mathrm{stat})}_{K,M}
&\coloneqq \beta_{K,M}
 + \sigma_K\sqrt{\frac{2\log(2/\delta)}{M}}.
\end{aligned}
\end{equation}
These components satisfy
$\Pr\bigl[|\widehat{\QFI}_{K,M}-\QFI|>
I^{(\mathrm{trunc})}_{K}+
I^{(\mathrm{stat})}_{K,M}\bigr]\le\delta$.
In the numerical stopping rule, $C_{\mathrm{trunc}}\mu^K$ is
represented by the inter-order change $d_K$
(Eq.~\eqref{eq:I-trunc}), and the stochastic sampling scale is
represented by the bootstrap empirical interval width $w_M$
(Eq.~\eqref{eq:I-stat}).
The bootstrap width does not by itself estimate the plug-in bias
$\beta_{K,M}$; any such bias must be made negligible by the
sample-count choice or included through separate calibration.
The adaptive controller monitors these two empirical quantities
componentwise and uses
$\max\{d_K,w_M\}$ as its allocation severity score.
This componentwise maximum rule keeps the two error sources visible
to the stopping decision.
The calibrated envelope above states the stronger condition needed
before the same structure can support a certified error bound.
The risk parameter $\delta$ in~\eqref{eq:indicator-theory} is
the same user-specified level as in
Theorem~\ref{thm:certified-stop}; the bootstrap interval
$w_M$ is a full-width empirical quantity at this nominal level,
not the calibrated radius $I^{(\mathrm{stat})}_{K,M}$ itself.

\subsection{Bootstrap Coverage Under Adaptive Stopping}
\label{sec:theory-bootstrap}

Bootstrap intervals are used in AKS-QFI to define the sampling-side
observable quantity $w_M$.
We characterize conditions under which they achieve nominal coverage
for the fixed-$K$ population target.

\paragraph{Fixed-$(K,M)$ bootstrap coverage.}
At any fixed $(K,M)$, let
$[\widehat{L}_{K,M},\,\widehat{U}_{K,M}]$ be the equal-tailed
percentile bootstrap interval at nominal level $1-\alpha$,
constructed from $B$ replicates of the estimator.

\begin{assumption}[Smooth estimator]
\label{ass:smooth}
For each fixed $K$, the estimator can be written as a plug-in
functional $T_K(\mathbb{P}_M)$ of the empirical shadow-measurement
law $\mathbb{P}_M$, with population value
$T_K(\mathbb{P})=\QFI_K$.
The map $T_K$ is Hadamard differentiable at $\mathbb{P}$ along
square-integrable directions, with influence function $\phi_K$
satisfying $\mathbb{E}[\phi_K^2] < \infty$.
\end{assumption}

\begin{proposition}[Fixed-$(K,M)$ bootstrap validity]
\label{prop:bootstrap-fixed}
Under Assumptions~\ref{ass:subgauss} and~\ref{ass:smooth},
as $M\to\infty$ with $K$ fixed and $B \to \infty$,
\begin{equation}
\Pr\!\left[
  \QFI_K \in [\widehat{L}_{K,M},\,\widehat{U}_{K,M}]
\right]
\to 1-\alpha.
\end{equation}
\end{proposition}

\begin{proof}
Assumption~\ref{ass:smooth} ensures the estimator is
asymptotically linear with influence function $\phi_K$, so
$\sqrt{M}(\widehat{\QFI}_{K,M} - \QFI_K)$ converges weakly
to $\mathcal{N}(0, \mathbb{E}[\phi_K^2])$.
Bootstrap consistency for smooth functionals of i.i.d.\
data~\cite{vaart1998asymptotic} then gives convergence of the
bootstrap distribution to the same Gaussian limit, from which
the percentile interval achieves nominal coverage.
\end{proof}

\begin{remark}[Smoothness of the plug-in estimator]
The benchmark estimator projects reconstructed density matrices
back to the density-matrix cone before evaluating the projected QFI\@.
This projection is smooth away from eigenvalue-clipping boundaries.
Near such boundaries, Assumption~\ref{ass:smooth} should be treated
as a local regularity condition rather than an automatic consequence
of the estimator.
\end{remark}

\begin{remark}[Coverage target is $\QFI_K$, not $\QFI$]
\label{rem:coverage-target}
Proposition~\ref{prop:bootstrap-fixed} covers the
Krylov-truncated population $\QFI_K$, not the target $\QFI$.
The bias $B_K = \QFI_K - \QFI$ is not captured by the
bootstrap interval unless $K$ is large enough that
$|B_K| \ll \sigma_K/\sqrt{M}$.
This is the formal mechanism behind the false-stop pathology:
at small $K$, the interval can be narrow and nominally valid
for $\QFI_K$ while remaining far from $\QFI$.
\end{remark}

\paragraph{Bootstrap coverage under adaptive stopping.}
In the AKS-QFI loop, $(K,M)$ are chosen adaptively, so
stopping introduces a data-dependent selection effect.

\begin{assumption}[Held-out stopping selection]
\label{ass:stop-indep}
Conditional on the selected terminal pair $(K_\tau,M_\tau)=(k,m)$,
the samples used to construct
$[\widehat{L}_{k,m},\widehat{U}_{k,m}]$ are independent of the
randomness used to decide that terminal pair and are distributed as
i.i.d.\ draws from the Krylov-shadow model at order $k$.
\end{assumption}

\begin{proposition}[Adaptive bootstrap coverage]
\label{prop:bootstrap-adaptive}
Under Assumptions~\ref{ass:subgauss},~\ref{ass:smooth},
and~\ref{ass:stop-indep}, let $\tau$ be the stopping time of
AKS-QFI\@.
Let $\mathcal{E}_{k,m}$ denote the event
$(K_\tau, M_\tau) = (k,m)$ for a deterministic pair $(k,m)$.
For any fixed $k$ and any sequence of sample counts $m\to\infty$
with $\Pr(\mathcal{E}_{k,m})>0$, conditionally on
$\mathcal{E}_{k,m}$,
\begin{multline}
\Pr\!\Bigl[
  \QFI_k \in [\widehat{L}_{k,m},\,\widehat{U}_{k,m}]
\\
  \,\Bigm|\, \mathcal{E}_{k,m}
\Bigr]
\to 1-\alpha
\quad \text{as } m\to\infty .
\end{multline}
\end{proposition}

\begin{proof}
For each $m$ in the sequence, Assumption~\ref{ass:stop-indep}
ensures that, conditional on $(K_\tau, M_\tau) = (k,m)$, the samples
used to construct the bootstrap interval have the same distribution as
a fixed-$(k,m)$ sample.
The result follows by applying
Proposition~\ref{prop:bootstrap-fixed} conditionally.
\end{proof}

\begin{remark}
Assumption~\ref{ass:stop-indep} holds when the stopping
decision uses a held-out batch independent of the
estimation batch.
In the same-sample rule used here, these quantities are computed
from the same batch as the estimate, so
Assumption~\ref{ass:stop-indep} is not literally satisfied.
Quantifying the resulting coverage distortion is a natural target
for a held-out-batch study.
\end{remark}

\paragraph{Finite-horizon sequential certificate.}
Proposition~\ref{prop:finite-horizon-heldout} in the main text gives
the finite-horizon held-out certificate associated with the
held-out component-aware interface.
It is stronger than the same-batch bootstrap statement above because
the adaptive exploration data can only nominate a candidate stop; the
reported estimate and statistical radius come from an independent
confirmation batch.
For the finite resource limits in the benchmark, we use Bonferroni
spending $\delta_j=\delta/J$; for an open-ended run, a summable
schedule such as $\delta_j=6\delta/(\pi^2j^2)$ gives the same union
bound with $J=\infty$.
The proposition requires a calibrated radius that includes both a
pre-registered Krylov truncation term and a confirmation-batch
sampling term.

\subsection{Proofs for Section~\ref{sec:stopping}}
\label{sec:theory-stopping}

This section provides the proof of
Theorem~\ref{thm:certified-stop} and the componentwise
interpretation in Remark~\ref{rem:false-stop-nec}, which are
stated in the main text.

\begin{definition}[Certified stopping claim]
\label{def:certified}
A stopping event at $(K,M)$ with tolerance $\varepsilon$ and
risk $\delta$ is \emph{certified} if
\begin{equation}
\Pr\!\left[
  \bigl|\widehat{\QFI}_{K,M} - \QFI\bigr| > \varepsilon
\right] \le \delta.
\end{equation}
\end{definition}

\begin{proof}[Proof of Theorem~\ref{thm:certified-stop}]
For a fixed pair $(K,M)$, Proposition~\ref{prop:error-bound}
gives
\begin{equation*}
\Pr\!\left[
  \bigl|\widehat{\QFI}_{K,M} - \QFI\bigr|
  >
  I^{(\mathrm{trunc})}_{K}
  + I^{(\mathrm{stat})}_{K,M}
\right] \le \delta.
\end{equation*}
Here the two components are calibrated upper bounds of the
form in Eq.~\eqref{eq:indicator-theory}.
Under conditions~\eqref{eq:suff-trunc-main}
and~\eqref{eq:suff-stat-main},
$I^{(\mathrm{trunc})}_{K}
+ I^{(\mathrm{stat})}_{K,M} \le \varepsilon$, so the event
$\{|\widehat{\QFI}_{K,M}-\QFI|>\varepsilon\}$ is contained
in the event on the left-hand side and
Eq.~\eqref{eq:certified-main} follows.
For an adaptively selected terminal pair, the same conclusion
requires that the sampling-side component remain valid at the
terminal time, for example through a held-out batch
or a time-uniform confidence sequence. This is the condition
stated in the theorem; a fully sequential finite-sample guarantee
for the bootstrap width would need such a construction.
\end{proof}

\paragraph{Derivation of Remark~\ref{rem:false-stop-nec}.}
If calibrated component bounds satisfy
conditions~\eqref{eq:suff-trunc-main}
and~\eqref{eq:suff-stat-main}, then
Theorem~\ref{thm:certified-stop} bounds the probability of a
post hoc error exceeding $\varepsilon$ by $\delta$.
Thus a single false success declaration is compatible with a valid
certificate through the allowed $\delta$-probability event.
A systematic false-success rate above the claimed risk level, or a
false success produced by observable scores that were never calibrated
as component radii, indicates that at least one component condition has
not been established at the tolerance scale.
In the benchmark trajectories this gap is visible on the
truncation side: the empirical gate can pass at the terminal
low-$K$ point, but higher-$K$ trajectories and the reference
value reveal a still-large truncation bias.

\paragraph{Practical gap.}
The theoretical constants $(C_{\mathrm{trunc}}, \mu, \sigma_K)$
in~\eqref{eq:indicator-theory} are unknown during a run and are
replaced in the numerical rule by measured quantities
$(d_K, w_M)$.
Theorem~\ref{thm:certified-stop} therefore serves as a
calibration principle for the present adaptive rule.
It explains why a stopping rule must track both the sampling
spread and the Krylov truncation side before it can support
a calibrated success claim.

\subsection{Stopping-Parameter Selection: Closed-Form Criteria}
\label{sec:theory-stopping-params}

The stopping parameters $K_{\min}^{\mathrm{stop}}$,
$M_{\min}^{\mathrm{stop}}$, and patience $P$ introduced in
Section~\ref{sec:stopping} are chosen by applying
Theorem~\ref{thm:certified-stop} as follows.

\paragraph{Minimum Krylov order.}
Under Assumption~\ref{ass:krylov-conv}, the condition
$I^{(\mathrm{trunc})}_K \le \varepsilon/2$ is satisfied when
\begin{equation}
\label{eq:kmin-theory}
K \ge K_{\min}^{\mathrm{stop}}(\varepsilon)
\coloneqq
\max\!\left\{1,\left\lceil
  \frac{\log(2C_{\mathrm{trunc}}/\varepsilon)}{\log(1/\mu)}
\right\rceil\right\}.
\end{equation}
When $(C_{\mathrm{trunc}}, \mu)$ are unknown, they are
estimated from the observed inter-order change sequence
$\{d_k\}_{k=1}^K$ by fitting an exponential decay
$d_k \approx \hat{C}\hat{\mu}^k$.
The default $K_{\min}^{\mathrm{stop}} = 4$ corresponds
to this empirical estimate under the tested noise family and
should be recalibrated for new problem classes.

\paragraph{Minimum sample count.}
When the plug-in bias term $\beta_{K,M}$ is negligible or has
already been included in the sampling-side calibration, the
condition $I^{(\mathrm{stat})}_{K,M} \le \varepsilon/2$
is satisfied by the stochastic threshold
\begin{equation}
\label{eq:mmin-theory}
M \ge M_{\min}^{\mathrm{stop}}(\varepsilon,\delta)
\coloneqq
\left\lceil
  \frac{8\sigma_K^2 \log(2/\delta)}{\varepsilon^2}
\right\rceil.
\end{equation}
More generally, if an upper bound $\bar{\beta}_K$ on the plug-in
bias is available over the relevant sampling range,
the corresponding sufficient condition, provided
$\bar{\beta}_K<\varepsilon/2$, is
\[
M \ge
\left\lceil
\frac{2\sigma_K^2\log(2/\delta)}
     {(\varepsilon/2-\bar{\beta}_K)^2}
\right\rceil .
\]
If only an $M$-dependent bound $\beta_{K,M}$ is available, the
required sample count is the smallest $M$ satisfying
$\beta_{K,M}+\sigma_K\sqrt{2\log(2/\delta)/M}\le\varepsilon/2$.
If the available bias bound is at least $\varepsilon/2$ over the
tested range, additional samples alone do not certify the stopping
claim without reducing or bounding the plug-in bias.
A high-confidence estimate of $\sigma_K^2$ can be obtained from the
per-shot sample variance for a linear or locally linearized estimator
at fixed $K$, or from $M$ times the variance of bootstrap replicates
of the plug-in estimate.
At $\varepsilon = 0.2$ and $\delta = 0.1$,
Eq.~\eqref{eq:mmin-theory} gives
$M_{\min}^{\mathrm{stop}}\approx 599\,\sigma_K^2$.
Thus an observed scale estimate $\sigma_K\approx0.56$ gives
$M_{\min}^{\mathrm{stop}}\approx188$, close to the benchmark
default of $128$, whereas the larger choice $\sigma_K=1$
would require a substantially larger threshold.
The benchmark value is calibrated to the observed sampling scale of the
tested noise family; transfer to other instances requires
recalibration.

\paragraph{Stopping patience.}
The following proposition establishes the theoretical basis
for the patience parameter $P$.

\begin{proposition}[Patience and false-stop suppression]
\label{prop:patience}
Let $p_{\mathrm{bad}} \in (0,1)$ be the per-eligible-step
probability that the componentwise stopping gates pass while
the true error exceeds $\varepsilon$.
Under the assumption that consecutive eligible-step readings
are independent, the probability of $P$ consecutive false
readings is $p_{\mathrm{bad}}^P$, which decreases
geometrically in $P$.
\end{proposition}

The independence assumption is an idealization: consecutive
readings share the same model state and partially overlapping
samples, so actual suppression may be weaker than
$p_{\mathrm{bad}}^P$.
Nevertheless, increasing $P$ still tightens the component-aware rule.
Setting $P = 2$ reduces a per-step false-stop probability of
$0.5$ to $0.25$, and $P = 3$ to $0.125$.

\paragraph{Joint stopping criterion.}
Combining the assumption-calibrated half-tolerance conditions, an
ideal assumption-calibrated rule would allow a run to declare
success only when
\begin{equation}
\label{eq:stop-criterion}
\begin{aligned}
K \ge K_{\min}^{\mathrm{stop}}(\varepsilon),
\\
M \ge M_{\min}^{\mathrm{stop}}(\varepsilon,\delta),
\\
I^{(\mathrm{trunc})}_{K} \le \frac{\varepsilon}{2},
\qquad
I^{(\mathrm{stat})}_{K,M} \le \frac{\varepsilon}{2}.
\end{aligned}
\end{equation}
The componentwise inequalities must hold for $P$ consecutive
eligible steps.
The practical rule has the same structure but
uses the measured quantities $(d_K,w_M)$ and user-level gates
$d_K\le\varepsilon$ and $w_M\le\varepsilon$, equivalently
$\max\{d_K,w_M\}\le\varepsilon$, with the finite-resource convention
that reaching $K_{\max}$ also passes the Krylov-side gate.
The factor-of-two distinction separates the theorem's calibrated
component bounds from the empirical gates used by the rule.
The benchmark therefore evaluates the resulting stopping decisions
post hoc against a reference value.
The default benchmark values
$(K_{\min}^{\mathrm{stop}}, M_{\min}^{\mathrm{stop}}, P)
= (4, 128, 2)$ stated in Eq.~\eqref{eq:stop-defaults}
approximate the theoretical
thresholds~\eqref{eq:kmin-theory}--\eqref{eq:mmin-theory}
under the tested noise family.
For a new problem class with different
$(C_{\mathrm{trunc}}, \mu, \sigma_K)$, these expressions
provide closed-form recalibration formulas that do not
require re-running the threshold ablation sweep.

\begin{remark}[Consistency with the ablation grid]
The ablation grid in Appendix~\ref{app:threshold-ablation-table}
sweeps
$(K_{\min}^{\mathrm{stop}}, M_{\min}^{\mathrm{stop}}, P)
\in \{2,4,6\} \times \{32,128,256\} \times \{1,2,3\}$.
Under criterion~\eqref{eq:stop-criterion}, settings below
the theoretical thresholds in $K$ or $M$ are predicted to
exhibit residual false stops when truncation or sampling
error is insufficiently controlled, consistent with the
threshold-grid structure in Figure~\ref{fig:threshold-ablation-pareto}.
\end{remark}

\section{Larger-system calibration-transfer tests at
\texorpdfstring{$n \in \{6,8\}$}{n in \{6,8\}} qubits}
\label{app:scaling}

This appendix reports stopping reliability results at
$n \in \{6,8\}$ qubits under the same small-resource protocol as
Section~\ref{sec:exp-setup}, but with the relative tolerance
$\varepsilon_{\mathrm{rel}}=0.05$
($p_{\mathrm{dep}}=0.03$, $(K_{\max},M_{\max})=(8,512)$,
10 independent replicates per setting).
These runs ask whether the default $n=4$ stopping parameters
$(K_{\min},M_{\min},P)=(4,128,2)$ transfer without recalibration.
They are separated from the main benchmark because, at these system
sizes, the resource limit $K_{\max}=8$ does not provide enough Krylov
resolution for the truncation bias to decay to the target tolerance.
The data therefore identify a calibration transfer point for the default
stopping rule; a fully recalibrated AKS-QFI scaling study would require
larger Krylov cutoffs or problem-specific truncation radii.

\paragraph{Calibration-transfer results.}
For both system sizes, width-only stopping exhibits
$\mathrm{FSR}=1.00$ at every noise level (95\% CI
$[0.72, 1.00]$, $n=10$ runs), with all runs terminating at
$K=2$, $M=16$.
The near-zero QFI estimates at $K=2$ have spuriously narrow
bootstrap CIs, satisfying the width-only stopping criterion
immediately regardless of $\varepsilon_{\mathrm{rel}}$.

For $n=6$, component-aware stopping reduces FSR at moderate-to-high
dephasing ($\mathrm{FSR}=0$ at $p_\phi \ge 0.12$), but
retains residual false stops at low dephasing
($\mathrm{FSR}=0.50$, 95\% CI $[0.24,0.76]$ at
$p_\phi=0$; $\mathrm{FSR}=0.20$, 95\% CI $[0.06,0.51]$
at $p_\phi=0.06$).
The component-aware runs have median effective terminal sample count
$384$ at $p_\phi=0$ and $512$ at the remaining tested
noise levels.
For $n=8$, the default $n=4$ component-aware calibration does not
suppress false stops:
$\mathrm{FSR}=1.00$ at all noise levels, with all runs declaring
success at $K=4$, $M=256$.
The component-aware rule's minimum Krylov order $K_{\min}=4$ is reached
before the truncation bias has decayed sufficiently, so
the empirical stopping criterion is satisfied at a biased estimate.
Because each setting uses only 10 replicates, the intervals are wide;
the qualitative conclusion is the systematic pattern of premature
success declarations under the un-recalibrated small-resource protocol.

\paragraph{Slow Krylov convergence at scale.}
Section~\ref{sec:exp-krylov} reports fitted decay rates
$\hat{\mu}(B_K) = 0.98$ ($n=6$) and $\hat{\mu}(B_K) = 1.00$
($n=8$), compared to $\hat{\mu}(B_K) = 0.89$ for $n=4$.
At $\hat{\mu} \approx 1$, the truncation bias
$|B_K| \approx C_{\mathrm{trunc}} \cdot \mu^K$ does not
meaningfully decrease within $K \le K_{\max}=8$.
The theoretical minimum Krylov order
(Eq.~\eqref{eq:kmin-theory}) depends inversely on
$\log(1/\mu)$; as $\mu \to 1$, $K_{\min}$ diverges,
meaning the calibrated Krylov-side condition cannot be met within a
fixed small resource limit.

Recalibrating the component-aware rule for $n \ge 6$ therefore requires
either increasing $K_{\max}$ substantially (with
corresponding increases in circuit depth and measurement cost),
or identifying noise regimes where $\mu$ is sufficiently
below 1.
The closed-form threshold expressions of
Appendix~\ref{sec:theory-stopping-params} provide the recalibration
tool once $\mu$ and $C_{\mathrm{trunc}}$ are estimated
from a short calibration run.
The present runs identify where the default calibration stops being
reliable; they do not rule out larger-system operation after
problem-specific Krylov and sampling calibration.

\bibliographystyle{quantum}
\bibliography{refs}

\end{document}